\documentclass[twoside,leqno,twocolumn]{article}
\usepackage{ltexpprt}
\usepackage{amsmath, amsfonts, amssymb}
\usepackage{multirow}
\usepackage{graphicx}

\usepackage[utf8]{inputenc}

\newtheorem{definition}{Definition}

\newtheorem{remark}{Remark}

\newtheorem{innercustomthm}{Theorem}
\newenvironment{customthm}[1]{\renewcommand\theinnercustomthm{#1}\innercustomthm}{\endinnercustomthm}

\newcommand{\F}{\mathbf{F}}
\newcommand{\Q}{\mathbf{Q}}
\newcommand{\U}{\mathbf{U}}
\newcommand{\LSF}{\mathcal{F}}
\newcommand{\RR}{\mathcal{V}}
\newcommand{\LSH}{\mathcal{H}}
\newcommand{\TLSF}{\mathbf{F}^{\otimes \tau}\!} 

\DeclareMathOperator*{\E}{\mathbb{E}}
\DeclareMathOperator{\poly}{poly}

\newcommand{\ip}[2]{\langle{#1},{#2}\rangle}
\newcommand{\va}{\mathbf{z}}

\newcommand{\x}{\mathbf{x}}
\newcommand{\y}{\mathbf{y}}
\newcommand{\z}{\mathbf{z}}

\newcommand{\real}{\mathbb{R}}
\newcommand{\sphere}[1]{\mathbb{S}^{#1}}
\newcommand{\norm}[1]{\left\lVert #1 \right\rVert}
\newcommand{\cube}[1]{\{0,1\}^{#1}}
\newcommand{\corrsub}[1]{\substack{(\x, \y) \\ {#1}\text{-correlated}}}
\newcommand*\mcap{\mathbin{\mathpalette\mcapinn\relax}}
\newcommand*\mcapinn[2]{\vcenter{\hbox{$\mathsurround=0pt\ifx\displaystyle#1\textstyle\else#1\fi\bigcap$}}}

\begin{document}
\pagenumbering{arabic} 
\title{\Large A Framework for Similarity Search with Space-Time Tradeoffs using Locality-Sensitive Filtering}
	
\author{Tobias Christiani\thanks{The research leading to these results has received funding 
	from the European Research Council under the European Union's Seventh Framework 
Programme (FP7/2007-2013) / ERC grant agreement no. [614331].} \\
\small \texttt{tobc@itu.dk}\\
\small IT University of Copenhagen}
\date{}

\maketitle
\pagenumbering{arabic} 
\begin{abstract}  {\small\baselineskip=9pt We present a framework for similarity search based on Locality-Sensitive Filtering~(LSF),
generalizing the Indyk-Motwani (STOC 1998) Locality-Sensitive Hashing~(LSH) framework to support space-time tradeoffs.    
Given a family of filters, defined as a distribution over pairs of subsets of space that satisfies certain locality-sensitivity properties, 
we can construct a dynamic data structure that solves the approximate near neighbor problem in $d$-dimensional space 
with query time $dn^{\rho_q + o(1)}$, update time $dn^{\rho_u + o(1)}$, 
and space usage $dn + n^{1 + \rho_u + o(1)}$ where $n$ denotes the number of points in the data structure.
The space-time tradeoff is tied to the tradeoff between query time and update time (insertions/deletions), 
controlled by the exponents $\rho_q, \rho_u$ that are determined by the filter family. \\   
Locality-sensitive filtering was introduced by Becker et al. (SODA 2016) together with a framework yielding a single, 
balanced, tradeoff between query time and space, further relying on the assumption of an efficient oracle for the filter evaluation algorithm.
We extend the LSF framework to support space-time tradeoffs 
and through a combination of existing techniques we remove the oracle assumption. \\
Laarhoven (arXiv 2015), building on Becker et al., introduced a family of filters with space-time tradeoffs 
for the high-dimensional unit sphere under inner product similarity and analyzed it for the important special case of random data.
We show that a small modification to the family of filters gives a simpler analysis that we use, 
together with our framework, to provide guarantees for worst-case data. 
Through an application of Bochner's~Theorem from harmonic analysis by Rahimi \& Recht (NIPS 2007), 
we are able to extend our solution on the unit sphere to $\real^d$ under the class of 
similarity measures corresponding to real-valued characteristic functions.
For the characteristic functions of $s$-stable distributions we obtain a solution 
to the $(r, cr)$-near neighbor problem in $\ell_s^d$-spaces with query and update exponents 
$\rho_q = \frac{c^s (1+\lambda)^2}{(c^s + \lambda)^2}$ and $\rho_u = \frac{c^s (1-\lambda)^2}{(c^s + \lambda)^2}$
where $\lambda \in [-1,1]$ is a tradeoff parameter. 
This result improves upon the space-time tradeoff of Kapralov (PODS 2015) and is shown to be optimal in the case of a balanced tradeoff, 
matching the LSH lower bound by O'Donnell et al.~\mbox{(ITCS 2011)} and a similar LSF lower bound proposed in this paper.
Finally, we show a lower bound for the space-time tradeoff on the unit sphere that matches Laarhoven's 
and our own upper bound in the case of random data.}
\end{abstract}
\section{Introduction}
Let $(X, D)$ denote a space over a set~$X$ equipped with a symmetric measure of dissimilarity~$D$ (a distance function in the case of metric spaces).
We consider the \emph{$(r, cr)$-near neighbor problem} first introduced by Minsky and Papert~\cite[p. 222]{minsky1969} in the~1960's.  
A solution to the $(r, cr)$-near neighbor problem for a set $P$ of $n$ points in $(X, D)$ 
takes the form of a data structure that supports the following operation: 
given a query point $\x \in X$, if there exists a data point $\y \in P$ such that $D(\x, \y) \leq r$ 
then report a data point $\y' \in P$ such that $D(\x, \y') \leq cr$.
In some spaces it turns out to be convenient to work with a measure of similarity rather than dissimilarity.
We use $S$ to denote a symmetric measure of similarity and define the \mbox{\emph{$(\alpha, \beta)$-similarity problem}}
to be the $(-\alpha, -\beta)$-near neighbor problem in~$(X, -S)$.

A solution to the $(r, cr)$-near neighbor problem can be viewed as a fundamental building block that yields solutions 
to many other similarity search problems such as the $c$-approximate \emph{nearest} neighbor problem \cite{indyk2004, har-peled2012}.
In particular, the $(r, cr)$-near neighbor problem is well-studied in $\ell_s^d$-spaces  
where the data points lie in~$\real^d$ and distances are measured by 
\mbox{$D(\x, \y) = \norm{\x - \y}_s = (\sum_{i=1}^{d}|x_i - y_i|^s)^{1/s}$}.
Notable spaces include the Euclidean space~$(\real^d, \norm{\cdot}_2)$, Hamming space~$(\{0,1\}^{d}, \norm{\cdot}_1)$, 
and the $d$-dimensional unit sphere~$\sphere{d} = \{ \x \in \real^d \mid \norm{\x}_2 = 1 \}$
under inner product similarity \mbox{$S(\x, \y) = \ip{\x}{\y} = \sum_{i=1}^{d}x_iy_i$}.

\paragraph{Curse of dimensionality}
All known solutions to the $(r, cr)$-near neighbor problem for $c = 1$ (the exact near neighbor problem) 
either suffer from a space usage that is exponential in $d$ or a query time that is linear in $n$ \cite{har-peled2012}.
This phenomenon is known as the ``curse of dimensionality'' and has been observed both in theory and practice.
For example, Alman and Williams \cite{alman2015} recently showed that the existence of an algorithm for determining 
whether a set of $n$ points in $d$-dimensional Hamming space contains a pair of points that are exact near neighbors 
with a running time strongly subquadratic in $n$ would refute the Strong Exponential Time Hypothesis (SETH) \cite{williams2004}.
This result holds even when $d$ is rather small, $d = O(\log n)$.
From a practical point of view, Weber et al. \cite{weber1998} showed that the performance of many of the tree-based approaches to similarity search 
from the field of computational geometry \cite{berg2008} degrades rapidly to a linear scan as the dimensionality increases.

\paragraph{Approximation to the rescue}
If we allow an approximation factor of $c > 1$ then there exist solutions to the $(r, cr)$-near neighbor problem 
with query time that is strongly sublinear in~$n$ and space polynomial in~$n$
where both the space and time complexity of the solution depends only polynomially on $d$.
Techniques for overcoming the curse of dimensionality through approximation were discovered 
independently by Kushilevitz et al. \cite{kushilevitz2000} and Indyk and Motwani \cite{indyk1998}.
The latter, classical work by Indyk and Motwani \cite{indyk1998, har-peled2012} introduced a general framework 
for solving the $(r, cr)$-near neighbor problem known as Locality-Sensitive Hashing (LSH). 
The introduction of the LSH framework has inspired an extensive literature (see e.g. \cite{andoni2008, wang2014} for surveys) 
that represents the state of the art in terms of solutions to the $(r, cr)$-near neighbor problem in high-dimensional spaces \cite{indyk1998, charikar2002, datar2004, panigrahy2006, andoni2006, andoni2008, andoni2009, andoni2015, kapralov2015, andoni2015data, becker2016, laarhoven2015}.

\paragraph{Hashing and filtering frameworks}
The LSH framework and the more recent LSF framework 
introduced by Becker et al. \cite{becker2016} produce data structures that solve the $(r, cr)$-near neighbor problem 
with query and update time $dn^{\rho + o(1)}$ and space usage $dn + n^{1 + \rho + o(1)}$. 
The LSH (LSF) framework takes as input a distribution over partitions (subsets) of space with the locality-sensitivity property 
that close points are more likely to be contained in the same part (subset) of a randomly sampled element from the distribution.
The frameworks proceed by constructing a data structure that associates each point in space with a number of memory locations
or ``buckets'' where data points are stored. 
During a query operation the buckets associated with the query point are searched by computing the distance to every data point in the bucket, 
returning the first suitable candidate.
The set of memory locations associated with a particular point is independent of whether an update operation or a query operation is being performed.
This symmetry between the query and update algorithm results in solutions to the near neighbor problem with a balanced space-time tradeoff.
The exponent~$\rho$ is determined by the locality-sensitivity properties of the family of partitions/hash functions (LSH) or subsets/filters (LSF) 
and is typically upper bounded by an expression that depends only on the aproximation factor $c$.
For example, Indyk and Motwani \cite{indyk1998} gave a simple locality-sensitive family of hash functions for Hamming space 
with an exponent of~$\rho \leq 1/c$. 
This exponent was later shown to be optimal by O'Donnell et al. \cite{odonnell2014} who gave a lower bound of $\rho \geq 1/c - o_{d}(1)$ 
in the setting where $r$ and $cr$ are small compared to $d$.
The advantage of having a general framework for similarity search lies in the reduction of the $(r, cr)$-near neighbor problem 
to the, often simpler and easier to analyze, problem of finding a locality-sensitive family of hash functions or filters for the space of interest.    

\paragraph{Space-time tradeoffs}
Space-time tradeoffs for solutions to the $(r, cr)$-near neighbor problem is an active line of research that can be 
motivated by practical applications where it is desirable to choose the tradeoff between query time and update time (space usage) 
that is best suited for the application and memory hierarchy at hand~\cite{panigrahy2006, lv2007, andoni2009, kapralov2015, laarhoven2015}. 
Existing solutions typically have query time $dn^{\rho_q + o(1)}$, update time (insertions/deletions) $dn^{\rho_u + o(1)}$, and use space $dn + n^{1 + \rho_u + o(1)}$ 
where the query and update exponents $\rho_q, \rho_u$ that control the space-time tradeoff depend on the approximation factor $c$ and on a tradeoff parameter $\lambda \in [-1,1]$.
This paper combines a number of existing techniques \cite{becker2016, laarhoven2015, dubiner2010} to provide a general framework for similarity search with space-time tradeoffs. 
The framework is used to show improved upper bounds on the space-time tradeoff in the well-studied setting of $\ell_s$-spaces and the unit sphere under inner product similarity.
Finally, we show a new lower bound on the space-time tradeoff for the unit sphere that matches an upper bound for random data on the unit sphere by Laarhoven \cite{laarhoven2015}.
We proceed by stating our contribution and briefly surveying the relevant literature in terms of frameworks, 
upper bounds, and lower bounds as well as some recent developments.
See table Table \ref{tab:results} for an overview.

\begin{table*}[t]
	\centering
	\renewcommand\arraystretch{2}
	\caption{Overview of data-independent locality-sensitive hashing (LSH) and filtering (LSF) results}
	\footnotesize
    \begin{tabular}{|c|c|c|c|c|} \hline
		\textbf{Reference} & \textbf{Setting} &  $\rho_q$ & $\rho_u$ \\ \hline\hline
		LSH \cite{indyk1998, har-peled2012}, LSF \cite{becker2016} &
		\multirow{2}{*}{$(X, D)$ or $(X, S)$}  
		& \multicolumn{2}{c|}{$\dfrac{\log (1/p)}{\log (1/q)}$} \\[3pt] \cline{1-1} \cline{3-4} 
		\textbf{Theorem \ref{thm:lsfvanilla}} & 
		& $\dfrac{\log (p_q /p_1)}{\log (p_q / p_2)}$ & $\dfrac{\log (p_u /p_1)}{\log (p_q / p_2)}$ \\[4pt] 
	\hline\hline 
	
	Cross-polytope LSH \cite{andoni2015} & $(\alpha, \beta)$-sim.\@ in $(\sphere{d},\ip{\cdot}{\cdot})$ & 
	\multicolumn{2}{c|}{
		$\left. \dfrac{1-\alpha}{1+\alpha}  \middle/  \dfrac{1-\beta}{1+\beta} \right.$
	} \\[3pt]  \hline
	
	Spherical cap LSF \cite{laarhoven2015} & $(\alpha, o_{d}(1))$-sim.\@ in $(\sphere{d},\ip{\cdot}{\cdot})$ & 
	$\dfrac{(1-\alpha^{1+\lambda})^2}{1-\alpha^2}$ &
	$\dfrac{(\alpha^{\lambda} -\alpha)^2}{1-\alpha^2}$ 
	\\[2pt]  \hline
	
	\textbf{Theorem \ref{thm:spherevanilla}} & $(\alpha, \beta)$-sim.\@ in $(\sphere{d},\ip{\cdot}{\cdot})$ & 
	$\left. \dfrac{(1-\alpha^{1+\lambda})^2}{1-\alpha^2}  \middle/  \dfrac{(1- \alpha^{\lambda} \beta)^2}{1-\beta^2} \right.$ &
	$\left. \dfrac{(\alpha^{\lambda} - \alpha)^2}{1-\alpha^2}  \middle/  \dfrac{(1-\alpha^{\lambda} \beta)^2}{1-\beta^2} \right.$ 
	\\[3pt]  \hline \hline

	Ball-carving LSH \cite{andoni2006} & \multirow{3}{*}{ $(r, cr)$-nn.\@ in $\ell_2^d$}  
	& \multicolumn{2}{c|}{$1/c^2$}   \\  \cline{1-1} \cline{3-4} 
	
	Ball-search LSH* \cite{kapralov2015} && 
	\scriptsize
	$\dfrac{c^2 (1 + \lambda)^2}{(c^2 + \lambda)^2 - c^2 (1+\lambda^2)/2 - \lambda^2}$ & 
	\scriptsize
	$\dfrac{c^2 (1 - \lambda)^2}{(c^2 + \lambda)^2 - c^2 (1+\lambda^2)/2 - \lambda^2}$ \\[2pt] \cline{1-1} \cline{3-4} 
	\footnotesize
	
	\textbf{Theorem \ref{thm:lsvanilla}}  && 
	$\dfrac{c^2 (1 + \lambda)^2}{(c^2 + \lambda)^2}$ & 
	$\dfrac{c^2 (1 - \lambda)^2}{(c^2 + \lambda)^2}$ \\[3pt] 
	\hline \hline

	Lower bound \cite{odonnell2014} & LSH in $\ell_2^d$ &
	\multicolumn{2}{c|}{$\geq 1/c^2$} \\  \hline
	\textbf{Theorem \ref{thm:lsflshvanilla}} & LSF in $\ell_2^d$ &
	\multicolumn{2}{c|}{$\geq 1/c^2$} \\  \hline
	Lower bound \cite{motwani2007, andoni2015lower} & LSH in $(\sphere{d},\ip{\cdot}{\cdot})$ &
	\multicolumn{2}{c|}{
		$\geq \dfrac{1-\alpha}{1+\alpha}$
	} \\[2pt]  \hline
	\textbf{Theorem \ref{thm:lowertradeoffvanilla}}, \cite{andoni2016} & LSF in $(\sphere{d},\ip{\cdot}{\cdot})$ &
	$\geq \dfrac{(1-\alpha^{1+\lambda})^2}{1-\alpha^2}$ &
	$\geq \dfrac{(\alpha^{\lambda} -\alpha)^2}{1-\alpha^2}$ 
	\\[2pt]  \hline
    \end{tabular}
\begin{minipage}{0.89\textwidth} 
\scriptsize
\vspace{0.6em}
\textsc{Table notes:} Space-time tradeoffs for dynamic randomized solutions to similarity search problems in the LSH and LSF frameworks 
with query time $dn^{\rho_q + o(1)}$, update time $dn^{\rho_u + o(1)} + dn^{o(1)}$ and space usage $dn + n^{1 + \rho_u + o(1)}$.
Lower bounds are for the exponents $\rho_q, \rho_u$ within their respective frameworks.
Here $\varepsilon > 0$ denotes an arbitrary constant and $\lambda \in [-1,1]$ controls the space-time tradeoff. 
We have hidden $o_{n}(1)$ terms in the upper bounds and $o_d(1)$ terms in the lower bounds.
\newline *Assumes $c^2 \geq (1+\lambda)^2 / 2 + \lambda + \varepsilon$.
\end{minipage}
\label{tab:results}
\end{table*}
\subsection{Contribution}
Before stating our results we give a definition of locality-sensitive filtering that supports asymmetry in the framework query and update algorithm, 
yielding space-time tradeoffs.
\begin{definition} \label{def:lsf}
	Let $(X, D)$ be a space and let $\LSF$ be a probability distribution over \mbox{$\{(Q, U) \mid Q \subseteq X, U \subseteq X \}$}.
	We say that $\LSF$ is $(r, cr, p_1, p_2, p_q, p_u)$-sensitive if for all points $\x, \y \in X$ 
	and $(Q, U)$ sampled randomly from $\LSF$ the following holds:
	\begin{itemize}
		\item[--] If $D(\x, \y) \leq r$ then $\Pr[\x \in Q, \y \in U] \geq p_1$.
		\item[--] If $D(\x, \y) > cr$ then $\Pr[\x \in Q, \y \in U] \leq p_2$.
		\item[--] $\Pr[\x \in Q] \leq p_q$ and $\Pr[\x \in U] \leq p_u$.
	\end{itemize}
	We refer to $(Q, U)$ as a filter and to $Q$ as the query filter and $U$ as the update filter.

\end{definition}

Our main contribution is a general framework for similarity search with space-time tradeoffs 
that takes as input a locality-sensitive family of filters.
\begin{theorem}\label{thm:lsfvanilla}
	Suppose we have access to a family of filters that is $(r, cr, p_1, p_2, p_q, p_u)$-sensitive.
	Then we can construct a fully dynamic data structure 
	that solves the $(r, cr)$-near neighbor problem with 
	query time $dn^{\rho_q + o(1)}$, update time $dn^{\rho_u + o(1)}$, and space usage $dn + n^{1 + \rho_u + o(1)}$ where
	$\rho_q = \frac{\log(p_{q} / p_{1}) }{\log(p_{q} / p_{2})}$ and $\rho_u = \frac{\log(p_{u} / p_{1})}{\log(p_{q} / p_{2}) }$. 
\end{theorem}  

We give a worst-case analysis of a slightly modified version of Laarhoven's \cite{laarhoven2015} filter family for the unit sphere 
and plug it into our framework to obtain the following theorem.
\begin{theorem}\label{thm:spherevanilla}
	For every choice of $0 \leq \beta < \alpha < 1$ and $\lambda \in [-1,1]$ 
	there exists a solution to the $(\alpha, \beta)$-similarity problem in $(\sphere{d}, \ip{\cdot}{\cdot})$ 
	that satisfies the guarantees from Theorem \ref{thm:lsfvanilla} with exponents  
	$\rho_q = \left. \frac{(1-\alpha^{1+\lambda})^2}{1-\alpha^2}  \middle/  \frac{(1- \alpha^{\lambda} \beta)^2}{1-\beta^2} \right.$ and  
	$\rho_u = \left. \frac{(\alpha^{\lambda} - \alpha)^2}{1-\alpha^2}  \middle/  \frac{(1-\alpha^{\lambda} \beta)^2}{1-\beta^2} \right.$. 
\end{theorem}

We show how an elegant and powerful application of Bochner's Theorem \cite{rudin1990} by Rahimi and Recht \cite{rahimi2007} 
allows us to extend the solution on the unit sphere to a large class of similarity measures, yielding as a special case solutions for $\ell_s$-space.
\begin{theorem} \label{thm:lsvanilla}
	For every choice of $c \geq 1$, $s \in (0, 2]$, and $\lambda \in [-1,1]$ 
	there exists a solution to the~$(r, cr)$-near neighbor problem in $\ell_s^d$ 
	that satisfies the guarantees from Theorem \ref{thm:lsfvanilla} with exponents
	$\rho_q = \frac{c^s (1 + \lambda)^2}{(c^s + \lambda)^2}$ and $\rho_u = \frac{c^s (1 - \lambda)^2}{(c^s + \lambda)^2}$.		
\end{theorem}
This result improves upon the state of the art for every choice of asymmetric query/update exponents $\rho_q \neq \rho_u$ \cite{panigrahy2006, andoni2006, andoni2009, kapralov2015}.
We conjecture that this tradeoff is optimal among the class of algorithms that \emph{independently of the data} determine which locations in memory to probe during queries and updates.
In the case of a balanced space-time tradeoff where we set $\rho_q = \rho_u$ our approach matches existing, optimal \cite{odonnell2014}, 
data-independent solutions in $\ell_s$-spaces~\cite{indyk1998, datar2004, andoni2006, nguyen2014}.

The LSF framework is very similar to the LSH framework, 
especially in the case where the filter family is symmetric ($Q = U$ for every filter in $\LSF$).
In this setting we show that the LSH lower bound by O'Donnell applies to the LSF framework as well \cite{odonnell2014}, 
confirming that the results of Theorem \ref{thm:lsvanilla} are optimal when we set $\rho_q = \rho_u$.
\begin{theorem}[informal] \label{thm:lsflshvanilla}
	Every filter family that is symmetric and $(r, cr, p_1, p_2, p_q, p_u)$-sensitive in $\ell_{s}^{d}$ 
	must have $\rho = \frac{\log (p_u /p_1)}{\log (p_q / p_2)} \geq 1/c^{s} - o_{d}(1)$ when $r = \omega_{d}(1)$ is chosen to be sufficiently small.
\end{theorem}

Finally we show a lower bound on the space-time tradeoff that can be obtained in the LSF framework.
Our lower bound suffers from two important restrictions. 
First the filter family must be regular, meaning that all query filters and all update filters are of the same size. 
Secondly, the size of the query and update filter cannot differ by too much.
\begin{theorem}[informal] \label{thm:lowertradeoffvanilla}
    Every regular filter family that is $((1-\alpha)d/2, (1-\beta)d/2, p_1, p_2, p_q, p_u)$-sensitive in $d$-dimensional Hamming space 
	with asymmetry controlled by $\lambda \in [-1,1]$ 
	cannot simultanously have that
	\mbox{$\rho_q < \frac{(1-\alpha^{1+\lambda})^2}{1-\alpha^2} - o_{d}(1)$} 
	and $\rho_u < \frac{(\alpha^{\lambda} -\alpha)^2}{1-\alpha^2} - o_{d}(1)$. 
\end{theorem}
Together our upper and lower bounds imply that the filter family of concentric balls in Hamming space is asymptotically optimal for random data.

\paragraph{Techniques}
The LSF framework in Theorem \ref{thm:lsfvanilla} relies on a careful combination of ``powering'' and ``tensoring'' techniques.
For positive integers $m$ and $\tau$ with $m \gg \tau$ the tensoring technique, 
a variant of which was introduced by Dubiner \cite{dubiner2010}, allows us to simulate a collection of $\binom{m}{\tau}$ filters 
from a collection of $m$ filters by considering the intersection of all $\tau$-subsets of filters.
Furthermore, given a point $\x \in X$ we can efficiently list the simulated filters that contain $\x$. 
This latter property is crucial as we typically need $\poly(n)$ filters to split our data 
into sufficiently small buckets for the search to be efficient.
The powering technique lets us amplify the locality-sensitivity properties of a filter family 
in the same way that powering is used in the LSH framework \cite{indyk1998, andoni2008, odonnell2014}.

To obtain results for worst-case data on the unit sphere we analyze a filter family based on standard normal projections using the same techniques as Andoni et al. \cite{andoni2015} together with existing tail bounds on bivariate Gaussians.
The approximate kernel embedding technique by Rahimi and Recht \cite{rahimi2007} is used to extend the solution on the unit sphere 
to a large class of similarity measures, yielding Theorem \ref{thm:lsvanilla} as a special case. 

The lower bound in Theorem \ref{thm:lsflshvanilla} relies on an argument of contradiction against the LSH lower bounds 
by O'Donnell \cite{odonnell2014} and uses a theoretical, inefficient, construction of a locality-sensitive family of hash functions 
from a locality-sensitive family of filters that is similar to the spherical LSH by Andoni et al. \cite{andoni2014}.

Finally, the space-time tradeoff lower bound from Theorem \ref{thm:lowertradeoffvanilla} is obtained through an application of an 
isoperimetric inequality by O'Donnell \cite[Ch. 10]{od2014} and is similar in spirit to the LSH lower bound by Motwani et al.~\cite{motwani2007}.
\subsection{Related work}
The LSH framework takes a distribution $\LSH$ over hash functions that partition space with the property 
that the probability of two points landing in the same partition is an increasing function of their similarity.
\begin{definition} \label{def:lsh}
	Let $(X, D)$ be a space and let $\LSH$ be a probability distribution over functions $h \colon X \to R$.
	We say that $\LSH$ is $(r, cr, p, q)$-sensitive if for all points $\x, \y \in X$ and $h$ sampled randomly from $\LSH$ the following holds:
	\begin{itemize}
		\item[--] If $D(\x, \y) \leq r$ then $\Pr[h(\x) = h(\y)] \geq p$.
		\item[--] If $D(\x, \y) > cr$ then $\Pr[h(\x) = h(\y)] \leq q$.
	\end{itemize}
\end{definition}
The properties of $\LSH$ determines a parameter $\rho < 1$ that governs the space and time complexity 
of the solution to the $(r, cr)$-near neighbor problem.
\begin{theorem}[LSH framework \cite{indyk1998, har-peled2012}]\label{thm:lsh}
	Suppose we have access to a $(r, cr, p, q)$-sensitive hash family.
	Then we can construct a fully dynamic data structure 
	that solves the $(r, cr)$-near neighbor problem with query time $dn^{\rho + o(1)}$, update time~$dn^{\rho + o(1)}$, 
	and with a space usage of $dn + n^{1 + \rho + o(1)}$ where \mbox{$\rho = \frac{\log (1 / p) }{\log (1 / q) }$}. 
\end{theorem}

The LSF framework by Becker et al. \cite{becker2016} takes a symmetric $(r, cr, p_1, p_2, p_q, p_u)$-sensitive 
filter family $\LSF$ and produces a data structure that solves the 
$(r, cr)$-near neighbor problem with the same properties as the one produced by the LSH framework 
where instead we have \mbox{$\rho = \frac{\log(p_{q}/p_{1})}{\log(p_{q}/p_{2})}$}.
In addition, the framework assumes access to an oracle that is able to efficiently list the relevant filters containing a point $\x \in X$
out of a large collection of filters.
The LSF framework in this paper removes this assumption, showing how to construct an efficient oracle as part of the framework. 

In terms of frameworks that support space-time tradeoffs, Panigrahy \cite{panigrahy2006} developed a 
framework based on LSH that supports the two extremes of the space-time tradeoff. 
In the language of Theorem \ref{thm:lsfvanilla}, 
Panigrahy's framework supports either setting $\rho_u = 0$ for a solution that uses near-linear space at the cost of a slower query time,
or setting $\rho_q = 0$ for a solution with query time $n^{o(1)}$ at the cost of a higher space usage.
To obtain near-linear space the framework stores every data point in $n^{o(1)}$ partitions 
induced by randomly sampled hash functions from a $(r, cr, p, q)$-sensitive LSH family $\LSH$.
In comparison, the standard LSH framework from Theorem \ref{thm:lsh} uses $n^{\rho}$ such partitions where $\rho$ is determined by $\LSH$.
For each partition induced by $h \in \LSH$ the query algorithm in Panigrahy's framework generates a number of random points $\z$ in a ball 
around the query point $\x$ and searches the parts of the partition $h(\z)$ that they hash to.
The query time is bounded by $n^{\hat{\rho} + o(1)}$ where $\hat{\rho} = \frac{I(h(\z) | \x, h)}{\log(1/q)}$ 
and $I(h(\z) | \x, h)$ denotes conditional entropy, i.e. the query time is determined by how hard it is to guess where 
$\z$ hashes to given that we know $\x$ and $h$.
Panigrahy's technique was used in a number of follow-up works that improve on solutions for specific spaces,
but to our knowledge none of them state a general framework with space-time tradeoffs \cite{lv2007, andoni2009, kapralov2015}.

\paragraph{Upper bounds}
As is standard in the literature we state results in $\ell_s$-spaces in terms of the properties of a solution to the $(r, cr)$-near neighbor problem.
For results on the unit sphere under inner product similarity $(\sphere{d}, \ip{\cdot}{\cdot})$
we instead use the $(\alpha, \beta)$-similarity terminology, defined in the introduction, as we find it to be 
cleaner and more intuitive while aligning better with the analysis.
The $\ell_s$-spaces, particularly $\ell_1$ and $\ell_2$, as well as $(\sphere{d}, \ip{\cdot}{\cdot})$ are some of 
most well-studied spaces for similarity search and are also widely used in practice \cite{wang2014}. 
Furthermore, fractional norms ($\ell_s$ for $s \neq 1, 2$) have been shown to perform better than the standard norms 
in certain use cases \cite{aggarwal2001} which motivates finding efficient solutions to the near neighbor problem in general $\ell_s$-space. 

In the case of a balanced space-time tradeoff the best data-independent upper bound for the $(r, cr)$-near neighbor problem in $\ell_s^d$ 
are solutions with an LSH exponent of $\rho = 1/c^s$ for $0 < s \leq 2$.
This result is obtained through a combination of techniques.
For \mbox{$0 < s \leq 1$} the LSH based on $s$-stable distributions by Datar et al. \cite{datar2004} 
can be used to obtain an exponent of $(1 + \varepsilon)/c^s$ for an arbitrarily small constant $\varepsilon > 0$.
For \mbox{$1 < s \leq 2$} the ball-carving LSH by Andoni and Indyk~\cite{andoni2006} for Euclidean space can be extended to $\ell_s$
using the technique described by Nguyen \cite[Section 5.5]{nguyen2014}.
Theorem \ref{thm:lsvanilla} matches (and potentially improves in the case of $0 < s < 1$) 
these results with a single unified technique and analysis that we find to be simpler.

For space-time tradeoffs in Euclidean space (again extending to $\ell_s$ for $1 < s < 2$) Kapralov \cite{kapralov2015},
improving on Panigrahy's results \cite{panigrahy2006} in Euclidean space and using similar techniques, obtains a solution with 
query exponent \mbox{$\rho_q = \frac{c^2 (1 + \lambda)^2}{(c^2 + \lambda)^2 - c^2 (1+\lambda^2)/2 - \lambda^2}$}
and update exponent \mbox{$\rho_u = \frac{c^2 (1 - \lambda)^2}{(c^2 + \lambda)^2 - c^2 (1+\lambda^2)/2 - \lambda^2}$} 
under the condition that $c^2 \geq (1+\lambda)^2 / 2 + \lambda + \varepsilon$ where $\varepsilon > 0$ is an arbitrary positive constant.
Comparing to our Theorem \ref{thm:lsvanilla} it is easy to see that we improve upon Kapralov's space-time tradeoff 
for all choices of $c$ and $\lambda$. 
In addition, Theorem \ref{thm:lsvanilla} represents the first solution to the $(r, cr)$-near neighbor problem in Euclidean space 
that for every choice of constant $c > 1$ obtains sublinear query time ($\rho_q < 1$) using only near-linear space ($\rho_u = 0$). 
Due to the restrictions on Kapralov's result he is only able to obtain sublinear query time for $c > \sqrt{3}$ 
when the space usage is restricted to be near-linear. 
It appears that our improvements can primarily be attributed to our techniques allowing a more direct analysis.
Kapralov uses a variation of Panigrahy's LSH-based technique of, depending on the desired space-time tradeoff, 
either querying or updating additional memory locations around a point $\x \in X$ in the partition induced by $h \in \LSH$.
For a query point $\x$ and a near neighbor $\y$ his argument for correctness is based on guaranteeing that both the query algorithm 
and update algorithm visit the part $h(\z)$ where $\z$ is a point lying between $\x$ and $\y$, 
possibly leading to a loss of efficiency in the analysis.
More details on the comparison of Theorem \ref{thm:lsvanilla} to Kapralov's result can be found in Appendix~\ref{app:kapralov}. 

In terms of space-time tradeoffs on the unit sphere, Laarhoven \cite{laarhoven2015} modifies a filter family 
introduced by Becker et al. \cite{becker2016} to support space-time tradeoffs,
obtaining a solution for random data on the unit sphere (the $(\alpha, \beta)$-similarity problem with $\beta = o_{d}(1)$) 
with query exponent \mbox{$\rho_q = \frac{(1 - \alpha^{1 + \lambda})^2}{1 - \alpha^2}$} 
and update exponent \mbox{$\rho_u = \frac{(\alpha^{\lambda} - \alpha)^2}{1 - \alpha^2}$}.
Theorem \ref{thm:spherevanilla} extends this result to provide a solution to the $(\alpha, \beta)$-similarity problem on the unit sphere
for every choice of $0 \leq \beta < \alpha < 1$.
This extension to worst case data is crucial for obtaining our results for $\ell_s$-spaces in Theorem \ref{thm:lsvanilla}. 
We note that there exist other data-independent techniques (e.g. Valiant \cite[Alg. 25]{valiant2015}) 
for extending solutions on the unit sphere to $\ell_2$, but they also require a solution for worst-case data on the unit sphere to work.

\paragraph{Lower bounds}
The performance of an LSH-based solution to the near neighbor problem in a given space that uses
a $(r, cr, p, q)$-sensitive family of hash functions $\LSH$ is summarized by the value of the exponent $\rho = \frac{\log(1/p)}{\log(1/q)}$. 
It is therefore of interest to lower bound $\rho$ in terms of the approximation factor $c$.
Motwani et al.~\cite{motwani2007} proved the first lower bound for LSH families in $d$-dimensional Hamming space. 
They show that for every choice of $c \geq 1$ then for some choice of $r$ it must hold that $\rho \geq 0.462/c$ as $d$ goes to infinity
under the assumption that $q$ is not too small ($q \geq 2^{-o(d)}$). 

As part of an effort to show lower bounds for data-dependent locality-sensitive hashing,
Andoni and Razenshteyn \cite{andoni2016datalower} strengthened the lower bound by Motwani et al.\ to $\rho \geq 1/(2c - 1)$ in Hamming space.
These lower bounds are initially shown in Hamming space and can then be extended to $\ell_s$-space and the unit sphere 
by the fact that a solution in these spaces can be used to yield a solution in Hamming space, 
contradicting the lower bound if $\rho$ is too small.
Translated to $(\alpha, \beta)$-similarity on the unit sphere, 
which is the primary setting for the lower bounds on LSF space-time tradeoffs in this paper,
the lower bound by Andoni and Razenshteyn shows that an LSH on the unit sphere must have $\rho \geq \frac{1-\alpha}{1+\alpha}$ 
which is tight in the case of random data~\cite{andoni2015}. 

The lower bound uses properties of random walks over a partition of Hamming space: 
A random walk starting from a random point $\x \in \cube{d}$ is likely to ``walk out'' 
of the the part identified by $h(\x)$ in the partition induced by $h$.
The space-time tradeoff lower bound in Theorem \ref{thm:lowertradeoffvanilla} relies on a similar argument 
that lower bounds the probability that a random walk starting from a subset $Q$ ends up in another subset $U$, 
corresponding nicely to query and update filters in the LSF framework. 

Using related techniques O'Donnell \cite{odonnell2014} showed tight LSH lower bounds for $\ell_s$-space of $\rho \geq 1/c^s$.
The work by Andoni et al.~\cite{andoni2006lower} and Panigrahy et al.~\cite{panigrahy2008, panigrahy2010} 
gives cell probe lower bounds for the $(r, cr)$-near neighbor problem, 
showing that in Euclidean space a solution with a query complexity of $t$ probes require space at least $n^{1 + \Omega(1/tc^2)}$. 
For more details on these lower bounds and how they relate to the upper bounds on the unit sphere see \cite{andoni2016, laarhoven2015}.

\paragraph{Data-dependent solutions}
The solutions to the $(r, cr)$-near neighbor problems considered in this paper are all \emph{data-independent}.
For the LSH and LSF frameworks this means that the choice of hash functions or filters used by the data structure,
determining the mapping between points in space and the memory locations that are searched during the query and update algorithm, 
is made without knowledge of the data.
Data-independent solutions to the $(r, cr)$-near neighbor problem for worst-case data have been the state of the art
until recent breakthroughs by Andoni et al.~\cite{andoni2014} and Andoni and Razenshteyn~\cite{andoni2015data}
showing improved solutions to the $(r, cr)$-near neighbor problem in Euclidean space using \emph{data-dependent} techniques.
In this setting the solution obtained by Andoni and Razenshteyn has an exponent of $\rho = 1/(2c^2 - 1)$ 
compared to the optimal data-independent exponent of $\rho = 1/c^2$.
Furthermore, they show that this exponent is optimal for data-dependent solutions in a restricted model \cite{andoni2016datalower}.

\paragraph{Recent developments} 
Recent work by Andoni et al.~\cite{andoni2016optimal}, done independently of and concurrently with this paper, 
shows that Laarhoven's upper bound for random data on the unit sphere can be combined with data-dependent techniques~\cite{andoni2015data} 
to yield a space-time tradeoff in Euclidean space with $\rho_u, \rho_q$ satisfying $c^2 \sqrt{\rho_q} + (c - 1)\sqrt{\rho_u} = \sqrt{2c^2 - 1}$.
This improves the result of Theorem~\ref{thm:lsvanilla} and matches the lower bound in Theorem~\ref{thm:lowertradeoffvanilla}.
In the same paper they also show a lower bound matching our lower bound in Theorem \ref{thm:lowertradeoffvanilla}.
Their lower bound is set in a more general model that captures both the LSH and LSF framework 
and they are able to remove some of the technical restrictions such as the filter family being regular that weaken the lower bound in this paper.
In spite of these results we still believe that this paper presents an important contribution 
by providing a general and simple framework with space-time tradeoffs as well as improved data-independent solutions 
to nearest neighbor problems in $\ell_s$-space and on the unit sphere. 
We would also like to point out the simplicity and power of using Rahimi and Recht's~\cite{rahimi2007} result 
to extend solutions on the unit sphere to spaces with similarity measures corresponding to real-valued characteristic functions, 
further described in Appendix~\ref{app:characteristic}.
\section{A framework with space-time tradeoffs}
We use a combination of powering and tensoring techniques to amplify the locality-sensitive properties of our initial filter family, 
and to simulate a large collection of filters that we can evaluate efficiently.
We proceed by stating the relevant properties of these techniques which we then combine to yield our Theorem \ref{thm:lsfvanilla}. 
\begin{lemma}[powering]\label{lem:powering}
	Given a $(r, cr, p_1, p_2, p_q, p_u)$-sensitive filter family $\LSF$ for $(X, D)$ 
	and a positive integer $\kappa$ define the family $\LSF^{\kappa}$ as follows:
	we sample a filter $F = (Q, U)$ from $\LSF^{\kappa}$ by sampling 
	$(Q_1, U_1), \dots, (Q_\kappa, U_\kappa)$ independently from $\LSF$
	and setting $(Q, U) = (\bigcap_{i=1}^{\kappa} Q_i, \bigcap_{i=1}^{\kappa} U_i)$. 
	The family $\LSF^{\kappa}$ is $(r, cr, p_{1}^{\kappa}, p_{2}^{\kappa}, p_{q}^{\kappa}, p_{u}^{\kappa})$-sensitive for $(X, D)$.
\end{lemma}

Let $\F$ denote a collection (indexed family) of $m$ filters and let $\Q$ and $\U$ denote the corresponding collections of query and update filters, 
that is, for $i \in \{1, \dots, m\}$ we have that $\F_i = (\Q_i, \U_i)$. 
Given a positive integer $\tau \leq m$ (typically $\tau \ll m$) we define $\TLSF$ to be the 
collection of filters formed by taking all the intersections of $\tau$-combinations of filters from $\F$, 
that is, for every $I \subseteq \{1, \dots, m \}$ with~$|I| = \tau$ we have that 
\begin{equation}
\F_{I}^{\otimes \tau}\! = \left(\mcap_{i \in I} \Q_{i}, \mcap_{i \in I} \U_{i}\right) 
\end{equation}
The following properties of the tensoring technique will be used to provide correctness, running time, 
and space usage guarantees for the LSF data structure that will be introduced in the next subsection.
We refer to the evaluation time of a collection of filters $\F$ as the time it takes, given a point $\x \in X$ 
to prepare a list of query filters $\Q(x) \subseteq \Q$ containing $\x$ and a list of update filters $\U(x) \subseteq \U$ containing $\x$
such that the next element of either list can be reported in constant time.
We say that a pair of points $(\x, \y)$ is contained in a filter $(Q, U)$ if~$\x \in Q$ and~$\y \in U$.
\begin{lemma}[tensoring]\label{lem:tensoring}
	Let $\LSF$ be a filter family that is $(r, cr, p_1, p_2, p_q, p_u)$-sensitive in $(X, D)$.
	Let $\tau$ be a positive integer and let $\F$ denote a collection of \mbox{$m = \lceil \tau / p_1 \rceil$} 
	independently sampled filters from~$\LSF$.
	Then the collection $\TLSF$ of $\binom{m}{\tau}$ filters has the following properties:
	\begin{itemize}
		\item[--] If $(\x, \y)$ have distance at most $r$ then with probability at least $1/2$ 
			there exists a filter in $\TLSF$ containing~$(\x, \y)$.
		\item[--] If $(\x, \y)$ have distance greater than $cr$ then the expected number of filters in $\TLSF$ 
			containing $(\x, \y)$ is at most~$p_{2}^{\tau}\binom{m}{\tau}$.
		\item[--] In expectation, a point $\x$ is contained in at most $p_{q}^{\tau}\binom{m}{\tau}$ query filters 
			and at most $p_{u}^{\tau}\binom{m}{\tau}$ update filters in $\TLSF$.
		\item[--] The evaluation time and space complexity of $\TLSF$ is dominated by the time 
			it takes to evaluate and store $m$ filters from $\LSF$. 
	\end{itemize}
\end{lemma}
\begin{proof}
	To prove the first property we note that there exists a filter in $\TLSF$ containing 
	$(\x, \y)$ if at least $\tau$ filters in $\F$ contain $(\x, \y)$.
	The binomial distribution has the property that the median is at least as great as the mean rounded down \cite{buhrman1980}.
	By the choice of $m$ we have that the expected number of filters in $\F$ containing $(\x, \y)$ is at least $\tau$ and the result follows. 
	The second and third properties follow from the linearity of expectation and the fourth is trivial.
\end{proof}
\subsection{The LSF data structure} \label{sec:lsfds}
We will introduce a dynamic data structure that solves the $(r, cr)$-near neighbor problem on a set of points $P \subseteq X$.
The data structure has access to a $(r, cr, p_1, p_2, p_q, p_u)$-sensitive filter family $\LSF$ in the sense that
it knows the parameters of the family and is able to sample, store, and evaluate filters from $\LSF$ in time $dn^{o(1)}$.

The data structure supports an initialization operation that initializes a collection of filters $\F$ 
where for every filter we maintain a (possibly empty) set of points from $X$.
After initialization the data structure supports three operations: \textsc{insert}, \textsc{delete}, and \textsc{query}.
The \textsc{insert} (\textsc{delete}) operation takes as input a point $\x \in X$ and adds (removes) the point from the set of 
points associated with each update filter in $\F$ that contains $\x$.
The \textsc{query} operation takes as input a point $\x \in X$. 
For each query filter in $\F$ that contains $\x$ we proceed by 
computing the dissimilarity $D(\x, \y)$ to every point $\y$ associated with the filter.
If a point $\y$ satisfying $D(\x, \y) \leq cr$ is encountered, then $\y$ is returned and the query algorithm terminates.
If no such point is found, the query algorithm returns a special symbol ``$\varnothing$'' and terminates.

The data structure will combine the powering and tensoring techniques in order to simulate the collection of filters $\F$ from two smaller collections:
$\F_1$ consisting of $m_1$ filters from $\LSF^{\kappa_1}$ and $\F_2$ consisting of $m_2$ filters from $\LSF^{\kappa_2}$.
The collection of simulated filters $\F$ is formed by taking all filters $(Q_1 \cap Q_2, U_1 \cap U_2)$ 
where $(Q_1, U_1)$ is a member of $\F_{1}^{\otimes \tau}\!$ and $(Q_2, U_2)$ is a member of $\F_2$.
It is due to the integer constraints on the parameter $\tau$ in the tensoring technique 
and the parameter $\kappa$ in the powering technique that we simulate our filters from two underlying collections 
instead of just one.
This gives us more freedom to hit a target level of amplification of the simulated filters which in turn
makes it possible for the framework to support efficient solutions for a wider range of parameters of LSF families. 

The initialization operation takes $\LSF$ and parameters $m_1, \kappa_1, \tau, m_2, \kappa_2$ and samples and stores $\F_1$ and $\F_2$.
The filter evaluation algorithm used by the insert, delete, and query operation takes a point $\x \in X$ and computes for $\F_1$ and $\F_2$, 
depending on the operation, the list of update or query filters containing $\x$.  
From these lists we are able to generate the list of filters in $\F$ containing~$\x$.

Setting the parameters of the data structure to guarantee correctness while balancing the contribution to the query time from 
the filter evaluation algorithm, the number of filters containing the query point, and the number of distant points examined, 
we obtain a partially dynamic data structure that solves the $(r, cr)$-near neighbor problem with failure probability $\delta \leq 1/2 + 1/e$.
Using a standard dynamization technique by Overmars and Leeuwen \cite[Thm. 1]{overmars1981} we obtain a fully dynamic data structure 
resulting in Theorem \ref{thm:lsfvanilla}. The details of the proof have been deferred to Appendix~\ref{app:framework}.
\section{Gaussian filters on the unit sphere} \label{sec:gaussianlsf}
In this section we show properties of a family of filters for the unit sphere $\sphere{d}$ under inner product similarity.
Later we will show how to make use of this family to solve the near neighbor problem in other spaces, including $\ell_s$ for $0 < s \leq 2$.
\begin{lemma} \label{lem:gaussianlsf}
	For every choice of $0 \leq \beta < \alpha < 1$, $\lambda \in [-1, 1]$, and $t > 0$ 
	let $\mathcal{G}$ denote the family of filters defined as follows:
	we sample a filter $(Q, U)$ from $\mathcal{G}$ by sampling $\z \sim \mathcal{N}^{d}(0,1)$ and setting
\begin{align}
	Q &= \{ \x \in \real^d \mid \ip{\x}{\z} > \alpha^{\lambda} t \}, \\ 
	U &= \{ \x \in \real^d \mid \ip{\x}{\z} >  t \}.
\end{align}
Then $\mathcal{G}$ is locality-sensitive 
on the unit sphere under inner product similarity with exponents
\small
\begin{align*}
	\rho_q &\leq \left.\left( \frac{(1-\alpha^{1+\lambda})^2}{1-\alpha^2} + 
\frac{\ln(2\pi(1+t/\alpha)^{2})}{t^2 / 2}\right) \middle/  \frac{(1-\alpha^{\lambda} \beta)^2}{1-\beta^2} \right., \\
\rho_u &\leq \left.\left( \frac{(\alpha^{\lambda} - \alpha)^2}{1-\alpha^2} + 
\frac{\ln(2\pi(1+t/\alpha)^{2})}{t^2 / 2}\right) \middle/  \frac{(1-\alpha^{\lambda} \beta)^2}{1-\beta^2} \right..
\end{align*}
\normalsize
\end{lemma}
Laarhoven's filter family~\cite{laarhoven2015} is identical to $\mathcal{G}$ 
except that he normalizes the projection vectors $\z$ to have unit length. 
The properties of $\mathcal{G}$ can easily be verified with a simple back-of-the-envelope analysis using two facts:
First, for a standard normal random variable $Z$ we have that $\Pr[Z > t] \approx e^{-t^{2}/2}$.
Secondly, the invariance of Gaussian projections $\ip{\x}{\z}$ to rotations, allowing us to analyze the projection of arbitrary points 
$\x, \y \in \sphere{d}$ with inner product $\ip{\x}{\y} = \alpha$ in a two-dimensional setting $\x = (1, 0)$ 
and $\y = (\alpha, \sqrt{1-\alpha^{2}})$ without any loss of generality.
The proof of Lemma~\ref{lem:gaussianlsf} as well as the proof of Theorem~\ref{thm:spherevanilla} has been deferred to Appendix~\ref{app:gaussian}.
\section{Space-time tradeoffs under kernel similarity} \label{sec:kernel}
In this section we will show how to combine the Gaussian filters for the unit sphere with kernel approximation techniques in order to solve the 
$(\alpha, \beta)$-similarity problem over $(\real^d, S)$ for the class of similarity measures of the form $S(\x, \y) = k(\x-\y)$ 
where $k \colon \real^d \times \real^d \to \real$ is a real-valued characteristic function \cite{ushakov1999}.
For this class of functions there exists a feature map $\psi$ into a (possibly infinite-dimensional) 
dot product space such that $k(\x, \y) = \ip{\psi(\x)}{\psi(\y)}$.
Through an elegant combination of Bochner's Theorem and Euler's Theorem, detailed in Appendix \ref{app:characteristic},
Rahimi and Recht \cite{rahimi2007} show how to construct approximate feature maps, 
i.e., for every $k$ we can construct a function $v$ with the property that $\ip{v(\x)}{v(\y)} \approx \ip{\psi(\x)}{\psi(\y)} = k(\x - \y)$.
We state a variant of their result for a mapping onto the unit sphere.
\begin{lemma} \label{lem:rahimirecht}
	For every real-valued characteristic function $k$ and every positive integer $l$ there exists
	a family of functions $\RR \subseteq \{ v \mid v \colon \real^{d} \to \sphere{l} \}$ such that 
	for every $\x, \y \in \real^d$ and $\varepsilon > 0$ we have that 
	\begin{equation}
	\Pr_{v \sim \RR}[|\ip{v(\x)}{v(\y)} - k(\x, \y)| \geq \varepsilon] \leq e^{-\Omega(l \varepsilon^2)}.
	\end{equation}
\end{lemma}
Theorem \ref{thm:characteristiclsf} in Appendix \ref{app:characteristic} shows that Theorem \ref{thm:spherevanilla} holds
with the space $(\sphere{d}, \ip{\cdot}{\cdot})$ replaced by $(\real^d, k)$.
\subsection{Tradeoffs in $\ell_s^d$-space} \label{sec:ls}
Consider the $(r, cr)$-near neighbor problem in $\ell_s^d$ for $0 < s \leq 2$. 
We solve this problem by first applying the approximate feature map from Lemma \ref{lem:rahimirecht} for the characteristic function 
of a standard $s$-stable distribution~\cite{zolotarev1986}, mapping the data onto the unit sphere, 
and then applying our solution from Theorem \ref{thm:spherevanilla} to solve the appropriate $(\alpha, \beta)$-similarity problem on the unit sphere.  
The characteristic functions of $s$-stable distributions take the following form:
\begin{lemma}[{Lévy \cite{levy1925}}]\label{lem:levy}
	For every positive integer $d$ and $0 < s \leq 2$ there exists a characteristic function 
	$k \colon \real^d \times \real^d \to [0,1]$ of the form
	\begin{equation}
	k(\x, \y) = k(\x-\y) = e^{-\norm{\x-\y}_{s}^{s}}.
	\end{equation}
\end{lemma}
A result by Chambers et al. \cite{chambers1976} shows how to sample efficiently from an $s$-stable distributions. 

To sketch the proof of Theorem \ref{thm:lsvanilla} we proceed by upper bounding the exponents 
$\rho_q$, $\rho_u$ from Theorem \ref{thm:spherevanilla} when applying Lemma \ref{lem:rahimirecht} 
to get $\alpha \geq e^{-r^s} - \varepsilon$ and $\beta \leq e^{-c^s r^s} - \varepsilon$.
We make use of the following standard fact (see e.g. \cite{savage1962}) that can be derived from the 
Taylor expansion of the exponential function: for $x \geq 0$ it holds that $1 - x \leq e^{-x} \leq 1 - x + x^{2}/2$.
Scaling the data points such that $r^s = o(1)$ and inserting the above values of $\alpha \approx 1 - r^s$ and $\beta \approx 1 - c^s r^s$ 
into the expressions for $\rho_q$, $\rho_u$ in Lemma \ref{lem:gaussianlsf} 
we can set parameters $t$ and $l$ such that Theorem \ref{thm:lsvanilla} holds. 
\section{Lower bounds}
We begin by stating the lower bound on the LSH exponent $\rho = \log(1/p) / \log(1/q)$ by O'Donnell et al.~\cite{odonnell2014}.
\begin{theorem}[O'Donnell et al. \cite{odonnell2014}]\label{thm:odlower}
	Fix $d \in \mathbb{N}$, $1 < c < \infty$, $0 < s < \infty$ and $0 < q < 1$. 
	Then for a certain choice of $r = \omega_{d}(1)$ and under the assumption that $q \geq 2^{-o(d)}$ 
	we have that every $(r, cr, p, q)$-sensitive family of hash functions for $\ell_s^d$ 
	must satisfy 
	\begin{equation}
		\rho = \frac{\log(1/p)}{\log(1/q)} \geq \frac{1}{c^{s}} - o_{d}(1).
	\end{equation}	
\end{theorem}
The following lemma shows how to use a filter family $\LSF$ to construct a hash family $\LSH$.
\begin{lemma}\label{lem:lsftolsh}
	Given a symmetric family of filters that is $(r, cr, p_1, p_2, p_q, p_u)$-sensitive in $(X, D)$
	we can construct a $(r, cr, p_{1}/(2p_{q}), p_{2}/p_{q})$-sensitive family of hash functions in $(X, D)$. 
\end{lemma}
\begin{proof}
	Given the filter family $\LSF$ we sample a random function $h$ from the hash family $\LSH$  
	taking an infinite sequence of independently sampled filters $(F_i)_{i=0}^{\infty}$ from~$\LSF$ 
	and setting $h(\x) = \min \, \{i \mid \x \in F_i\}$.
	The probability of collision is given by
	\begin{equation}	
	\Pr_{h \sim \LSH}[h(\x) = h(\y)] = \frac{\Pr_{F \sim \LSF}[\x \in F \land \y \in F]}{\Pr_{F \sim \LSF}[\x \in F \lor \y \in F]}
	\end{equation}
	and the result follows from the properties of $\LSF$.
\end{proof}
If the LSH family in Lemma \ref{lem:lsftolsh} had $p = p_{1}/p_{q}$ and $q = p_{2}/p_{q}$ then the lower bound would follow immediately.
We apply the powering technique from Lemma \ref{lem:powering} to the underlying filter family 
in order make the factor $2$ in $p_{1}/(2p_{q})$ disappear in the statement of $\rho$ as $d$ tends to infinity.
\begin{customthm}{1.4} \label{thm:lower1}
	Every symmetric $(r, cr, p_1, p_2, p_q, p_u)$-sensitive filter family $\LSF$ for $\ell_s^d$ must satisfy the lower bound
	of Theorem \ref{thm:odlower} with $p = p_{1}/p_{q}$ and $q = p_{2}/p_{q}$.
\end{customthm}
\begin{proof}
	Given a family $\LSF$ that satisfies the requirements from Theorem \ref{thm:odlower} there exists an integer $\kappa = \omega_{d}(1)$
	such the hash family $\LSH$ that results from applying Lemma \ref{lem:lsftolsh} to the powered family $\LSF^{\kappa}$ 
	also satisfies the requirements from Theorem \ref{thm:odlower}.
	The constructed family $\LSH$ is $(r, cr, p, q)$-sensitive for $p = (1/2) \cdot (p_{1}/p_{q})^{\kappa}$ and $q = (p_{2}/p_{q})^{\kappa}$.
	By our choice of $\kappa$ we have that $\log(1/p)/\log(1/q) = \log(p_{q}/p_{1})/\log(p_{q}/p_{2}) + o_{d}(1)$ 
	and the lower bound on $\log(1/p)/\log(1/q)$ from Theorem \ref{thm:odlower} applies.
\end{proof}
\subsection{Asymmetric lower bound}
The lower bound is based on an isoperimetric-type inequality that holds for randomly correlated points in Hamming space.
We say that the pair of points $(\x, \y)$ is $\alpha$-correlated if $\x$ is a random point in $\cube{d}$ 
and $\y$ is formed by taking $\x$ and independently flipping each bit with probability $(1-\alpha)/2$.    
We are now ready to state O'Donnell's generalized small-set expansion theorem.
Notince the similarity to the value of $p_1$ for the Gaussian filter family described in 
Section~\ref{sec:gaussianlsf} and Appendix~\ref{app:gaussian}.
\begin{lemma}[{\cite[p. 285]{od2014}}] \label{lem:odhyper}
	For every \mbox{$0 \leq \alpha < 1$}, $-1 \leq \lambda \leq 1$, and \mbox{$Q, U \subseteq \cube{d}$} 
	satisfying that \mbox{$|Q|/2^{d} = (|U|/2^{d})^{\alpha^{2\lambda}}$} we have 
	\begin{equation}
		\Pr_{\corrsub{\alpha}}[\x \in Q, \y \in U] \leq (|U|/2^{d})^\frac{1 + \alpha^{2\lambda} - 2 \alpha^{1 + \lambda}}{1 - \alpha^{2}}.
	\end{equation}
\end{lemma}
The argument for the lower bound assumes a regular $(r, cr, p_1, p_2, p_q, p_u)$-sensitive filter family $\LSF$ for Hamming space
where we set $r = (1-\alpha)d/2$ and \mbox{$cr = (1-\beta)d/2$} for some choice of $0 < \beta < \alpha < 1$.
We then proceed by deriving constraints on $p_1$, $p_2$, $p_q$, $p_u$, and minimize $\rho_q$ and $\rho_u$ subject to those constrains.
The proof of Theorem \ref{thm:lower2} is provided in Appendix~\ref{app:lowertradeoffproof}.
\begin{customthm}{1.5} \label{thm:lower2}
Fix $0 < \beta < \alpha < 1$. 
Then for every regular $((1-\alpha)d/2, (1-\beta)d/2, p_1, p_2, p_q, p_u)$-sensitive filter family in $d$-dimensional Hamming space 
with and $|Q|/2^d = (|U|/2^d)^{\alpha^{2\lambda}}$ where $\lambda$ satisfies 
$\alpha + 2\sqrt{\ln(d) / d} \leq \alpha^{\lambda} \leq 1/(\alpha - 2\sqrt{\ln(d)/d})$ it must hold that
\begin{align*}
	\rho_q &= \frac{\log(p_q / p_1)}{\log(p_q / p_2)} \geq \frac{(1-\alpha^{1+\lambda})^{2}}{1-\alpha^{2}} - o_{d}(1), \\ 
	\rho_u &= \frac{\log(p_u / p_1)}{\log(p_q / p_2)} \geq \frac{(\alpha^{\lambda} - \alpha)^{2}}{1-\alpha^{2}} - o_{d}(1) 
\end{align*}
when $p_q$ is set to minimize $\rho_q$ and we assume that $|U|/2^{d} \geq 2^{-o_{d}(1)}$.
\end{customthm}
\section{Open problems}
An important open problem is to find simple and practical data-dependent solutions to the $(r, cr)$-near neighbor problem.
Current solutions, the Gaussian filters in this paper included, suffer from $o(1)$ terms in the exponents that decrease very slowly in $n$.
A lower bound for the unit sphere by Andoni et al. \cite{andoni2015} indicates that this might be unavoidable.

Another interesting open problem is finding the shape of provably exactly optimal filters in different spaces.
In the random data setting in Hamming space, this problem boils down to maximizing the number of pairs of points 
below a certain distance threshold that is contained in a subset of the space of a certain size.
This is a fundamental problem in combinatorics that has been studied by among others \cite{kahn1988}, but a complete answer remains elusive.
The LSH and LSF lower bounds~\cite{motwani2007, odonnell2014, andoni2016datalower}, 
along with classical isoperimetric inequalities such as Harper's Theorem 
and more recent work summarized in the book by O'Donnell \cite{od2014} hints that the answer is somewhere between a subcube and a generalized sphere.

A recent result by Chierichetti and Kumar \cite{chierichetti2015} characterizes the set of 
transformations of LSH-able similarity measures as the set of probability-generating functions. 
This seems to have deep connections to result of this paper that uses characteristic functions that allow well-known kernel transformations. 
It seems possible that this paper can be viewed as a semi-explicit construction of their result, 
or that their result can described as an application of Bochner's Theorem.  
\newpage
\section*{Acknowledgment}
I would like to thank Rasmus Pagh for suggesting the application of Rahimi \& Recht's result \cite{rahimi2007} 
and the MinHash-like~\cite{broder1998} connection between LSF and LSH used in Theorem \ref{thm:lower1}. 
I would also like to thank Gregory Valiant and Udi Wieder for useful discussions about locality-sensitive filtering 
and the analysis of boolean functions during my stay at Stanford.
Finally, I would like to thank the Scalable Similarity Search group at the IT University of Copenhagen for feedback during the writing process, 
and in particular Martin Aumüller for pointing out the importance of a general framework for locality-sensitive filtering with space-time tradeoffs. 
\appendix

\section{Framework} \label{app:framework}
We state a version of Theorem \ref{thm:lsfvanilla} where the parameters of the filter family are allowed to depend on $n$.
\begin{customthm}{1.1} \label{thm:lsf}
	Suppose we have access to a filter family that is $(r, cr, p_1, p_2, p_q, p_u)$-sensitive.
	Then we can construct a fully dynamic data structure that solves the $(r, cr)$-near neighbor problem. 
	Assume that $1/p_{1}$, $1 / \log(p_{q}/p_{2})$, and $\exp(\log(1/p_{1})/\log(\min(p_{q},p_{u})/p_{1}))$ are $n^{o(1)}$, then the data structure has
	\begin{itemize}
		\item[--] query time $dn^{\rho_q + o(1)}$,
		\item[--] update time $n^{\rho_u + o(1)} + dn^{o(1)}$, 
		\item[--] space usage $n^{1 + \rho_u + o(1)} + dn + dn^{o(1)}$
	\end{itemize}
	where
	\begin{equation}
		\rho_q = \frac{\log p_{q} / p_{1} }{\log p_{q} / p_{2} }, \quad  
		\rho_u = \frac{\log p_{u} / p_{1} }{\log p_{q} / p_{2} }. 
	\end{equation}
\end{customthm}

To prove Theorem \ref{thm:lsf}, we begin by setting the parameters mentioned in the description of 
the LSF data structure in Section \ref{sec:lsfds}. 
\begin{align}
\kappa_{1} &= \left \lceil \frac{ \min(\rho_q, \rho_u) \log n}{\log(1/p_{1})} \right \rceil \\
\tau  &= \left \lfloor \frac{\log n}{\kappa_1 \log(p_{q}/p_{2})} \right \rfloor \leq \frac{\log(1/p_1)}{\log(\min(p_{q}, p_{u})/p_{1})} \\
m_{1} &= \lceil \tau/p_{1}^{\kappa_{1}}  \rceil \\
\kappa_{2} &= \max(0, \lceil \log(n) / \log(p_{q}/p_{2}) \rceil - \tau \kappa_1) \\
m_{2} &= \lceil 1/p_{1}^{\kappa_{2}} \rceil
\end{align}
We will now briefly explain the reasoning behind the parameter settings.
Begin by observing that the powering and tensoring techniques both amplify the filters from $\LSF$. 
Let $m = \binom{m_{1}}{\tau} \cdot m_{2}$ denote the number of simulated filters in our collection $\F$ 
and let $a = \tau\kappa_1 + \kappa_2$ be an integer denoting the number of times each filter has been amplified.
Ignoring the time it takes to evaluate the filters, the query time is determined by the sum of the number of filters that contain a query point 
and the number of distant points associated with those filters that the query algorithm inspects.
The expected number of activated filters is given by $mp_{q}^{a}$ while the worst case expected number of distant points to be inspected by the query algorithm is given by $nmp_{2}^{a}$.
Balancing the contribution to the query time from these two effects (ignoring the $O(d)$ factor from distance computations) 
results in a target value of $a = \lceil \log(n) / \log(p_{q}/p_{2}) \rceil$.
Compared to having an oracle that is able to list the filters from a collection that contains a point, 
there is a small loss in efficiency from using the tensoring technique due to the increase in the number of filters required to guarantee correctness.
The parameters of the LSF data structure are therefore set to minimize the use of tensoring such that
the time spent evaluating our collection of filters roughly matches the minimum of the query and update time.

Consider the initialization operation of the LSF data structure with the parameters setting from above.
We have that $\kappa_2 \leq \kappa_1$ implying that $m_2 = O(m_1)$. 
The initialization time and the space usage of the data structure prior to any insertions 
is dominated by the time and space used to sample and store the filters in $\F_1$.
By the assumption that a filter from $\LSF$ can be sampled in $O(d)$ operations and stored using $O(d)$ words, 
we get a space and time bound on the initialization operation of
\begin{equation}
	O(d \kappa_1 m_1) = O\left(dn^{\min(\rho_q, \rho_u)} \frac{p_{1} \log(n)}{\log(p_{q}/p_{2})}\right).
\end{equation}
Importantly, this bound also holds for the running time of the filter evaluation algorithm, that is, 
the preprocessing time required for constant time generation of the next element in the list of filters in $\F$ containing a point.
In the following analysis of the update and query time we will temporarily ignore the running time of the filter evaluation algorithm.

The expected time to insert or delete a point is dominated by the number of update filters in $\F$ that contains it.
The probability that a particular update filter in $\F$ contains a point is given by $p_{u}^{a}$.
Using a standard upper bound on the binomial coefficient we get that $m = O(e^{\tau}/p_{1}^{a})$ resulting in an expected update time of
\begin{equation}
O(m p_{u}^{a} + d) = O(n^{\rho_{u}} (p_{u}/p_{1}) e^{\tau} + d).
\end{equation}
In the worst case where every data point is at distance greater than $cr$ from the query point and has collision probablity $p_2$
the expected query time can be upper bounded by
\begin{equation}
O(mp_{q}^{a} + dnmp_{2}^{a}) = O(n^{\rho_q}e^{\tau}(p_{q}/p_{1} + d)). 
\end{equation}
With respect to the correctness of the query algorithm, if a near neighbor $\y$ to the query point $\x$ exists in $P$,
then it is found by the query algorithm if $(\x, \y)$ is contained in a filter in $\F_{1}^{\otimes \tau}\!$ as well as in a filter in $\F_2$.
By Lemma \ref{lem:tensoring} the first event happens with probability at least $1/2$ and by the choice of $m_2$, 
the second event happens with probability at least $1 - (1-p_{1}^{\kappa_2})^{p_{1}^{\kappa_2}} \geq 1 - 1/e$.
From the independence between $\F_1$ and $\F_2$ we can upper bound the failure probability $\delta \leq (1/2)(1+1/e)$.
This completes the proof of Theorem \ref{thm:lsf}.
\section{Gaussian filters} \label{app:gaussian}
In this section we upper and lower bound the probability mass in the tail of the bivariate standard normal distribution
when the correlation between the two standard normals is at most $\beta$ (upper bound) or at least $\alpha$ (lower bound).
We make use of the following upper and lower bounds on the univariate standard normal as well as an upper bound for the multivariate case.
\begin{lemma}[Follows Szarek \& Werner \cite{szarek1999}] \label{lem:univariatebounds}
	Let $Z$ be a standard normal random variable. Then, for every $t \geq 0$ we have that
	\small
	\begin{equation*}
		\frac{1}{\sqrt{2\pi}}\frac{1}{t+1}e^{-t^{2}/2} \leq \Pr[Z \geq t] \leq \frac{1}{\sqrt{\pi}}\frac{1}{t+1}e^{-t^{2}/2}. 
	\end{equation*}
	\normalsize
\end{lemma}
\begin{lemma}[Lu \& Li \cite{lu2009}]\label{lem:luli}
Let $\z$ be a $d$-dimensional vector of i.i.d. standard normal random variables 
and let $D \subset \real^d$ be a closed convex domain that does not contain the origin.
Let~$\Delta$ denote the Euclidean distance to the unique closest point in $D$, then we have that 
\begin{equation}
	\Pr[\z \in D] \leq e^{-\Delta^{2}/2}.
\end{equation}
\end{lemma}
\begin{lemma}[Tail upper bound] \label{lem:upper}
	For $\alpha, \lambda, t, \beta$ satisfying \mbox{$0 < \alpha < 1$}, $-1 \leq \lambda \leq 1$, $t > 0$, and $-1 < \beta < \alpha$ 
	every pair of standard normal random variables $(X, Y)$ with correlation $\beta' \leq \beta$ satisfies 
\begin{equation}
	\Pr[X \geq t \land Y \geq \alpha^{\lambda} t] \leq e^{-\Delta^{2}/2}
\end{equation}
where $\Delta^2 = (1 + \frac{(\alpha^{\lambda} - \beta)^2}{1-\beta^2})t^2$.    
\begin{proof}
For $\beta' = -1$ the result is trivial. 
For values of $\beta'$ in the range \mbox{$-1 < \beta' \leq \beta$} we use the $2$-stability of the normal distribution to analyze a tail bound for 
$(X, Y)$ in terms of a Gaussian projection vector $\z = (Z_1, Z_2)$ applied to unit vectors \mbox{$\x, \y \in \real^2$}.
That is, we can define \mbox{$X = \ip{\z}{\x}$} and $Y = \ip{\z}{\x}$ for some appropriate choice of $\x$ and $\y$.
Without loss of generality we set $\x = (1, 0)$ and note that for $\E[XY] = \beta'$ we must have that \mbox{$\y = (\beta', \sqrt{1-\beta'^2})$}. 
If we consider the region of $\real^2$ where $\z$ satisfies $X \geq t \land Y \geq \alpha^{\lambda} t$ 
we get a closed domain $D$ defined by $\z = (Z_1, Z_2)$ such that $Z_1 \geq t$ 
and \mbox{$Z_2 \geq (\alpha^{\lambda} t - \beta' Z_1)/(\sqrt{1-\beta'^2})$}.
The squared Euclidean distance from the origin to the closest point in $D$ at least 
$\Delta^2$ as can be seen by the fact that $\Delta^2$ decreasing in $\beta$.
Combining this observation with Lemma \ref{lem:luli} we get the desired result.
\end{proof}
\end{lemma}
\begin{lemma}[Tail lower bound] \label{lem:lower}
For $\alpha, \lambda, t$ satisfying $0 < \alpha < 1$, $-1 \leq \lambda \leq 1$, 
and $t > 0$ every pair of standard normal random variables $(X, Y)$ with correlation $\alpha' \geq \alpha$ satisfies 
\begin{equation}
	\Pr[X \geq t \land Y \geq \alpha^{\lambda} t] \geq \frac{e^{-\Delta^{2}/2}}{2\pi(1+ t/\alpha)^{2}}    
\end{equation}
where $\Delta^2 = (1 + \frac{(\alpha^{\lambda} - \alpha)^2}{1-\alpha^2})t^2$.    
\end{lemma}
\begin{proof}
	For $\alpha' = 1$ the result follows directly from Lemma \ref{lem:univariatebounds}. 
	For $\alpha' < 1$ we use the trick from the proof of Lemma \ref{lem:upper} and define $X = \ip{\z}{\x}$ and $Y = \ip{\z}{\x}$ 
	where $\x = (1, 0)$ and $\y = (\alpha, \sqrt{1-\alpha^2})$ and $\z = (Z_1, Z_2)$ is a vector of two i.i.d. standard normal random variables.
	This allows us to rewrite the probability as follows:
	\begin{align*}
		&\Pr[Z_1 \geq t \land \alpha Z_1 + \sqrt{1-\alpha^2}Z_2 \geq \alpha^{\lambda} t] \\
		&= \Pr[Z_1 \geq t]\Pr[\alpha Z_1 + \sqrt{1-\alpha^2}Z_2 \geq \alpha^{\lambda} t \mid Z_1 \geq t] \\
		&\geq \Pr[Z_1 \geq t]\Pr[\alpha t + \sqrt{1-\alpha^2}Z_2 \geq \alpha^{\lambda} t]  
	\end{align*}
	By the restrictions on $\alpha$ and $\lambda$ we have that \mbox{$(\alpha^{\lambda} - \alpha)t / \sqrt{1-\alpha^2} \leq t/\alpha$}.
	The result follows from applying the lower bound from Lemma \ref{lem:univariatebounds} and noting that the bound is increasing in $\alpha$.
\end{proof}
\subsection{Space-time tradeoffs on the unit sphere}
Summarizing the bound from the previous section, the family $\mathcal{G}$ from Lemma \ref{lem:gaussianlsf} satisfies that
\begin{align}
	p_1 &\geq \frac{e^{-(1 + \frac{(\alpha^{\lambda} - \alpha)^2}{1-\alpha^2})t^2/2}}{2\pi(1+ t/\alpha)^{2}} \\
	p_2 &\leq e^{-(1 + \frac{(\alpha^{\lambda} - \beta)^2}{1-\beta^2})t^2/2} \\
	p_q &\leq e^{-\alpha^{2\lambda}t^{2}/2} \\
	p_u &\leq e^{-t^{2}/2}.
\end{align}

We combine the Gaussian filters with Theorem \ref{thm:lsf} to show that we can solve the $(\alpha, \beta)$-similarity problem efficiently 
for the full range of space/time tradeoffs, even when $\alpha, \beta$ are allowed to depend on $n$, 
as long as the gap $\alpha - \beta$ is not too small. 
\begin{customthm}{1.2}\label{thm:sphere}
	For every choice of $0 \leq \beta < \alpha < 1$ and $\lambda \in [-1, 1]$ 
	we can construct a fully dynamic data structure that solves the $(\alpha, \beta)$-similarity problem in $(\sphere{d}, \ip{\cdot}{\cdot})$.
	Suppose that $\alpha - \beta \geq (\ln n)^{-\zeta}$ for some constant $\zeta < 1/2$, 
	that satisfies the guarantees from Theorem \ref{thm:lsfvanilla} with exponents  
	$\rho_q = \left. \frac{(1-\alpha^{1+\lambda})^2}{1-\alpha^2}  \middle/  \frac{(1- \alpha^{\lambda} \beta)^2}{1-\beta^2} \right.$ and  
	$\rho_u = \left. \frac{(\alpha^{\lambda} - \alpha)^2}{1-\alpha^2}  \middle/  \frac{(1-\alpha^{\lambda} \beta)^2}{1-\beta^2} \right.$. 
\end{customthm}
\begin{proof}
	Assuming that $\alpha - \beta \geq (\ln n)^{-\zeta}$ there exists a constant $\varepsilon > 0$ where by setting
	the parameter $t$ of $\mathcal{G}$ such that $t^2 / 2 = \frac{1-\beta^2}{(1-\alpha^{\lambda} \beta)^2}(\ln n)^{\varepsilon}$
	the family of filters satisfies the assumptions in Theorem \ref{thm:lsf} while guaranteeing 
	that the second term in $\rho_q$ and $\rho_u$ from Lemma \ref{lem:gaussianlsf} are $o(1)$.
\end{proof}
\begin{remark}
	Theorem \ref{thm:sphere} aims for simplicity and generality while allowing $\alpha$ and $\beta$ to depend on $n$. 
	For specific values of $\alpha, \beta, \lambda$ it is easy to find better bounds on the probabilties 
	(e.g. the bounds by Savage \cite{savage1962}) and to adjust $t$ in Lemma \ref{lem:gaussianlsf} 
	to avoid powering (setting $\kappa_{1} = 1,  \kappa_{2} = 0$) in the LSF framework.
\end{remark}
\section{Approximate feature maps, characteristic functions, and Bochner's Theorem} \label{app:characteristic}
We begin by defining what a characteristic function is and listing some properties that are useful for our application.
More information about characteristic functions can be found in the books by Lukacs~\cite{lukacs1970} and Ushakov~\cite{ushakov1999}.
\begin{lemma}[\cite{lukacs1970, ushakov1999}] \label{lem:characteristicproperties}
	Let $Z$ denote a random variable with distribution function $\mu$.
	Then the characteristic function $k(\Delta)$ of $Z$ is defined as 
		\begin{equation}
			k(\Delta) = \int_{-\infty}^{\infty}\mu(t)e^{i \Delta t} dt
		\end{equation}
	and it has the following properties:
	\begin{itemize}
		\item[-] A distribution function is symmetric if and only if its characteristic function is real and even.
		\item[-] Every characteristic function $k(\Delta)$ is uniformly continuous, has $k(0) = 1$, and $|k(\Delta)| \leq 1$ for all real $\Delta$.
		\item[-] Suppose that $k(\Delta)$ denotes the characteristic function of an absolutely continuous distribution
			     then \mbox{$\lim_{\Delta \rightarrow \infty}|k(\Delta)| = 0$}.
		\item[-] Let $X$ and $Y$ be independent random variables
			with characteristic functions $k_{X}$ and $k_{Y}$. 
			Then the characteristic function of $Z = (X,Y)$ is given by 
			$k(x, y) = k_{X}(x) k_{Y}(y)$. 
	\end{itemize}
\end{lemma}

Bochner's Theorem reveals the relation between characteristic functions and the class of real-valued functions 
$k(\x, \y)$ that admit a feature space representation $k(\x, \y) = \ip{\phi(\x)}{\phi(\y)}$
\begin{theorem}[Bochner's Theorem \cite{rudin1990}]
	A function $k : \real^d \times \real^d \to [0,1]$ is positive definite if and only if it can be written on the form
	\begin{equation}
		k(\x, \y) = \int_{\real^d} \mu(\va) e^{i \ip{\va}{\x-\y}}  d\va
	\end{equation}
	where $\mu$ is the probability density function of a symmetric distribution.
\end{theorem}

Rahimi \& Recht's \cite{rahimi2007} family of approximate feature maps $\RR$ is constructed from Bochner's Theorem by making use of Euler's Theorem as follows:
\begin{align*}
	k(\x, \y) &= \int_{\real^d} \mu(\va) e^{i \ip{\va}{\x-\y}}  d\va \\
			  &= \int_{\real^d} \mu(\va)( \cos(\ip{\va}{\x-\y}) + i \sin(\ip{\va}{\x-\y})) d\va \\
	&= \E_{\va}[\cos(\ip{\va}{\x-\y})]  \\ 
	&= \E_{\va, b}[\cos(\ip{\va}{\x-\y}) + \cos(\ip{\va}{\x} + \ip{\va}{\y}+ 2b)]  \\ 
	&= 2\E_{\va, b}[\cos(\ip{\va}{\x} + b) \cdot \cos(\ip{\va}{\y} + b)]. 
\end{align*}
Where the third equality makes use of the fact that $k(\x, \y)$ is real-valued to remove the complex part of the integral
and the fifth equality uses that $2\cos(x)\cos(y) = \cos(x + y) + \cos(x - y)$.

Now that we have an approximate feature map onto the sphere for the class of shift-invariant kernels, 
we will take a closer look at what functions this class contains, and what their applications are for similarity search. 
Given an arbitrary similarity function, we would like to be able to determine whether it is indeed a characteristic function.
Unfortunately, there are no known simple techniques for answering this question in general.
However, the machine learning literature contains many applications of different shift-invariant kernels \cite{scholkopf2002} 
and many common distributions have real characteristic functions (see Appendix B in \cite{ushakov1999} for a long list of examples).
Characteristic functions are also well studied from a mathematical perspective \cite{lukacs1970, ushakov1999},
and a number of different necessary and sufficient conditions are known.
A classical result by Pólya \cite{polya1949} gives simple sufficient conditions for a function to be a characteristic function.
Through the vectorization property from Lemma \ref{lem:characteristicproperties}, 
Pólya's conditions directly imply the existence of a large class of similarity measures on $\real^d$ that can fit into the above framework. 
\begin{theorem}[Pólya \cite{polya1949}]
	Every even continuous function $k : \real \to \real$ satisfying the properties
	\begin{itemize}
		\item[-] $k(0) = 1$ 
		\item[-] $\lim_{\Delta \to \infty}k(\Delta) = 0$ 
		\item[-] $k(\Delta)$ is convex for $\Delta > 0$ 
	\end{itemize}
	is a characteristic function.
\end{theorem}
Based on the results of Section \ref{sec:ls} one could hope for the existence of characteristic functions of the form 
$k(\Delta) = e^{-|\Delta|^s}$ for $s > 2$ but it is known that such functions cannot exist \cite[Theorem D.8]{benyamini1998}.
Furthermore, Marcinkiewicz \cite{marcinkiewicz1939} shows that a function of the form $k(\Delta) = \exp(-\poly(\Delta))$ 
cannot be a characteristic function if the degree of the polynomial is greater than two.

We state a more complete, constructive version of Lemma \ref{lem:rahimirecht} as well as the proof here.
\begin{lemma} \label{lem:rahimirecht2}
	Let $k$ be a real-valued characteristic function with associated distribution function $\mu$ and let $l$ be a positive integer.
	Consider the family of functions $\RR \subseteq \{ v \mid v \colon \real^{d} \to \sphere{l} \}$ 
	where a randomly sampled function $v$ is defined by, independently for $j = 1,\dots,l$, 
	sampling $\va$ from $\mu$ and $b$ uniformly on $[0, 2\pi]$, letting $\hat{v}(\x)_{j} = \sqrt{(2/l)}\cos(\ip{\va}{\x} + b)$
	and normalizing $v(\x)_{j} = \frac{\hat{v}(\x)}{\norm{\hat{v}(\x)}}$.
	The family $\RR$ has the property that for every $\x, \y \in \real^d$ and $\varepsilon > 0$ we have that 
	\begin{equation}
		\Pr_{v \sim \RR}[|\ip{v(\x)}{v(\y)} - k(\x, \y)| \geq \varepsilon] \leq 6e^{-l \varepsilon^2 / 128}.
	\end{equation}
\end{lemma}
\begin{proof}
	Since $l \cdot \hat{v}(\x)_j \hat{v}(\y)_j$ is bounded between $2$ and $-2$, and we have independence for different values of $j$, 
	Hoeffding's inequality \cite{hoeffding1963} can be applied to show that for every fixed pair of points $\x, \y$
	and $\hat{\varepsilon} > 0$ it holds that
	\begin{equation}
		\Pr[|\ip{\hat{v}(\x)}{\hat{v}(\y)} - k(\x, \y)| \geq \hat{\varepsilon}] \leq 2e^{-l\hat{\varepsilon}^{2}/8}.
	\end{equation}
	From the properties of characteristic functions we have that $k(\x, \x) = 1$ and $k(\x, \y) \leq 1$.
	The bound on the deviation of 
	\begin{equation}
		\ip{v(\x)}{v(\y)} = \frac{\ip{\hat{v}(\x)}{\hat{v}(\y)}}{\sqrt{\ip{\hat{v}(\x)}{\hat{v}(\x)}\ip{\hat{v}(\y)}{\hat{v}(\y)}}}
	\end{equation}
	from $k(\x, \y)$ follows from setting $\hat{\varepsilon} = \varepsilon/4$ 
	and using a union bound over the probabilities that the deviation of one of the inner products is too large.
\end{proof}
Combining the approximate feature map onto the unit sphere with Theorem \ref{thm:spherevanilla} we obtain the following:
\begin{theorem} \label{thm:characteristiclsf}
	Let $k \colon \real^d \to \real$ be a characteristic function and define the similarity measure $S(\x, \y) = k(\x-\y)$.
	Assume that we have access to samples from the distribution associated with $k$, 
	then Theorem \ref{thm:sphere} holds with $(\sphere{d}, \ip{\cdot}{\cdot})$ replaced by $(\real^d, S)$.
\end{theorem}
\begin{proof}
	According to Lemma \ref{lem:rahimirecht2}, we can set $l = n^{o(1)}$ to obtain a map $v \colon \real^d \to \sphere{l}$
	such that the the inner product on $\sphere{l}$ preserves the pairwise similarity between $n^{O(1)}$ 
	points with additive error $\varepsilon = o(1)$.
	This map has a space and time complexity of $O(dl) = dn^{o(1)}$.
	After applying $v$ to the data we can solve the $(\alpha, \beta)$-similarity problem on $(\real^d, k(\x-\y))$ by solving the 
	$(\alpha - \varepsilon, \beta + \varepsilon)$-similarity problem on $(\sphere{d}, \ip{\cdot}{\cdot})$.
	We can use Theorem \ref{thm:sphere} to construct a fully dynamic data structure for solving this problem, 
	adjusting the parameter $\lambda$ so that it continues to lie in the admissible range.
	The space and time complexities follow.
\end{proof}
\section{Proof of Theorem 1.5} \label{app:lowertradeoffproof}
Consider $\rho_q = \frac{\log(p_{q} / p_{1})}{\log(p_{q} / p_{2})}$.
Subject to the (implicit) LSF constraint that $p_q, p_u > p_1 > p_2 > 0$ we see that $\rho_q$ is minimized by setting 
$p_q, p_2$ as small as possible and $p_1$ as large as possible.
We will therefore derive lower bounds on $p_q, p_2$ and an upper bound on $p_1$.
For every value of $p_1$ and $p_2$ we minimize $\rho_q, \rho_u$ by choosing $p_q$ as small as possible.

For a random point $\x \in \cube{d}$ it must hold that $\Pr_{\LSF}[\x \in Q] = |Q|/2^{d}$.
This implies the existence of a fixed point $\y \in \cube{d}$ with the property that $\Pr_{\LSF}[\y \in Q] \geq |Q|/2^{d}$.
A regular filter family must therefore satisfy that $p_q \geq |Q|/2^{d}$ and $p_u \geq |U|/2^{d}$. 
Let $\lambda$ be defined as in Lemma \ref{lem:odhyper} then by a similar argument we have that $p_2 \geq (U/2^d)^{1+\alpha^{2\lambda}}$. 

In order to upper bound $p_1$ we make use of Lemma \ref{lem:odhyper} together with the following lemma that follows directly 
from an application of Hoeffding's inequality~\cite{hoeffding1963}.
\begin{lemma}\label{lem:correlationconcentration}
For every $0 < \varepsilon < (1-\alpha)/2$ we have that
\begin{equation*}
	\Pr_{\substack{(\x, \y) \\ \alpha+\varepsilon\text{-correlated}}}
	\left[\frac{1}{d}\sum_{i=1}^{d}(-1)^{\x_i}(-1)^{\y_i} \leq \alpha\right] \leq e^{-\varepsilon^{2}d/2}.
\end{equation*}
\end{lemma}
In the following derivation, assume that $\alpha, \varepsilon$ satisfies $0 < \varepsilon < (1-\alpha)/2$, 
let $\x, \y$ denote randomly $(\alpha + \varepsilon)$-correlated vectors in $\cube{d}$, 
and assume that $\alpha + \varepsilon \leq \alpha^{\lambda} \leq 1/(\alpha + \varepsilon)$, then
\begin{align*}
	&(|U|/2^{d})^\frac{1 + \alpha^{2\lambda} - 2 \alpha^{\lambda} (\alpha + \varepsilon)}{1 - (\alpha+\varepsilon) ^{2}} 
	\geq \Pr[\x \in Q, \y \in U] \\
	&\quad \geq \Pr[\x \in Q, \y \in U \mid \ip{\x}{\y} \geq \alpha] \Pr[\ip{\x}{\y} \geq \alpha] \\
	&\quad \geq  p_1 (1 - e^{-\varepsilon^{2} d/2})   
\end{align*}
Summarizing the bounds:
\begin{align*}
	p_1 &\leq \frac{(|U|/2^{d})^\frac{1 + \alpha^{2\lambda} - 2\alpha^{\lambda}(\alpha+\varepsilon)}{1-\alpha^{2}}}{1-e^{-\varepsilon^{2}d/2}} \\ 
	p_2 &\geq (|U|/2^{d})^{1 + \alpha^{2\lambda}}\\ 
	p_q &\geq |Q|/2^{d} \\
	p_u &\geq |U|/2^{d}.
\end{align*}

When minimizing $\rho_q$ we have that \mbox{$\log(p_{q}/p_{2}) = -\log(|U|/2^{d})$}.
Setting $\varepsilon = 2\sqrt{\ln(d) /d}$ results in 
$\log(1/p_1) \geq -\frac{1+\alpha^{2\lambda} - 2\alpha^{\lambda}(\alpha + \varepsilon)}{1-\alpha^{2}}\log(|U|/2^{d}) - O(1/d^{2})$.
Putting things together:
\begin{align*}
	\frac{\log(p_q / p_1)}{\log(p_q / p_2)} &\geq 
	-\frac{\alpha^{2\lambda} \log(|U|/2^{d})}{\log(|U|/2^d)} \\ 
	&\quad +  
	\frac{\frac{1+\alpha^{2\lambda} - 2\alpha^{\lambda}(\alpha + \varepsilon)}{1-\alpha^{2}}\log(|U|/2^{d})
	+ O(1/d^{2})}{\log(|U|/2^d)} \\
	&= \frac{(1 - \alpha^{1 + \lambda})^{2} -  2\alpha^{\lambda}\varepsilon}{1-\alpha^{2}} + \frac{ O(1/d^{2})}{\log(|U|/2^d)} \\
	&= \frac{(1 - \alpha^{1 + \lambda})^{2}}{1-\alpha^{2}} - O(\sqrt{\log(d)/d}).
\end{align*}
The derivation of the lower bound for $\rho_u$ is almost the same and the resulting expression is
\begin{equation}
	\frac{\log(p_u / p_1)}{\log( p_q / p_2)} \geq \frac{(\alpha^{\lambda} - \alpha)^{2}}{1-\alpha^{2}} - O(\sqrt{\log(d)/d}).
\end{equation}
\section{Comparison to Kapralov} \label{app:kapralov}
Kapralov uses $\alpha$ to denote a parameter controlling the space-time tradeoff 
for his solution to the $(r, cr)$-near neighbor problem in Euclidean space.
For every choice of tradeoff parameter $\alpha \in [0,1]$, 
assuming that \mbox{$c^2 \geq 3(1-\alpha)^2 - \alpha^2 + \varepsilon$} for arbitrarily small constant $\varepsilon > 0$, 
Kapralov \cite{kapralov2015} obtains query and update exponents
\begin{align}
	\rho_q &= \frac{4(1-\alpha)^2}{c^2 + (1-\alpha)^2 - 3\alpha^2}, \\
	\rho_u &= \frac{4\alpha^2}{c^2 + (1-\alpha)^2 - 3\alpha^2}.
\end{align}
We convert Kapralov's notation to our own by setting $\lambda = 1 - 2\alpha$.
To compare, Kapralov sets $\alpha = 0$ for near-linear space and we set $\lambda = 1$.
We want to write Kapralov's exponents on the form
\begin{equation}
	\rho_q = \frac{c^2(1+\lambda)^2}{(c^2 + \lambda)^2 + x}, \quad \rho_u = \frac{c^2(1-\lambda)^2}{(c^2 + \lambda)^2 + x}
\end{equation}
for some $x$ that we will proceed to derive.
We have that $(1-\alpha)^2 = (1+\lambda)^{2}/4$ and $\alpha^2 = (1-\lambda)^{2}/4$.
Multiplying the numerator and denominator in Kapralov's exponents by $c^2$ we can write Kapralov's exponents as
\begin{align}
	\rho_q &= \frac{c^2(1+\lambda)^2}{c^4 + c^2 (1+\lambda)^2 / 4 - 3c^2 (1-\lambda)^2 / 4}, \\
	\rho_u &= \frac{c^2(1-\lambda)^2}{c^4 + c^2 (1+\lambda)^2 / 4 - 3c^2 (1-\lambda)^2 / 4}.
\end{align}
We have that 
\begin{align*}
	x &= c^4 + c^2 (1+\lambda)^2 / 4 - 3c^2 (1-\lambda)^2 / 4 - (c^2 + \lambda)^2\\ 
	  &= -c^2(1+\lambda^2)/2 - \lambda^2. 
\end{align*}
For every choice of $\lambda \in [-1,1]$, and under the assumption that $c^2 \geq (1+\lambda)^2 / 2 + \lambda + \varepsilon$ 
for an arbitrarily small constant $\varepsilon > 0$, this allows us to write Kapralov's exponents as 
\begin{align}
	\rho_q &= \frac{c^2(1+\lambda)^2}{(c^2 + \lambda)^2 - c^2(1+\lambda^2)/2 - \lambda^2}, \\
	\rho_u &= \frac{c^2(1-\lambda)^2}{(c^2 + \lambda)^2 - c^2(1+\lambda^2)/2 - \lambda^2}.
\end{align}
To compare Kapralov's result against our own for search in $\ell_s$-spaces we consider the exponents from Theorem \ref{thm:lsvanilla},
ignoring additive $o(1)$ terms:
\begin{equation}
		\rho_q = \frac{c^s (1 + \lambda)^2}{(c^s + \lambda)^2}, \quad 
		\rho_u = \frac{c^s (1 - \lambda)^2}{(c^s + \lambda)^2}.		
\end{equation}
Setting $\lambda = 1$ we obtain a data structure that uses near-linear space and we get a query exponent $\rho_q = 16/25$ 
while Kapralov obtains an exponent of $\rho_q = 16/20$, ignoring $o(1)$ terms.
At the other end of the tradeoff, setting $\lambda = -1$, we get a data structure with query time $n^{o(1)}$ and update exponent $\rho_u = 16/9$ 
while Kapralov gets an update exponent of $\rho_u = 4$, again ignoring additive $o(1)$ terms.

The assumption made by Kapralov that $c^2 \geq (1+\lambda)^2 / 2 + \lambda + \varepsilon$ means that in the case of a 
near-linear space data structure ($\lambda = 1$) sublinear query time can only be obtained for $c > \sqrt{3}$.
In contrast, Theorem \ref{thm:lsvanilla} gives sublinear query time for every constant $c > 1$.
\section{Details about dynamization and the model of computation}
In order to obtain fully dynamic data structures we apply a powerful dynamization 
result of Overmars and Leeuwen \cite{overmars1981} for decomposable searching problems.
Their result allows us to turn a partially dynamic data structure into a fully dynamic data structure, 
supporting arbitrary sequences of queries and updates, at the cost of a constant factor in the space and running time guarantees. 
Suppose we have a partially dynamic data structure that solves the $(r, cr)$-near neighbor problem on a set of $n$ points.
By partially dynamic we mean that, after initialization on a set $P$ of $n$ points, 
the data structure supports $\Theta(n)$ updates without changing the query time by more than a constant factor.
Let $T_{q}(n)$, $T_{u}(n)$, and $T_{c}(n)$ denote the query time, update time, and construction time of such a data structure containing $n$ points.
Then, by Theorem 1 of Overmars and Leeuwen \cite{overmars1981}, there exists a fully dynamic version of the data structure 
with query time $O(T_{q}(n))$ and update time $O(T_{u}(n) + T_{c}(n)/n)$ that uses only a constant factor additional space.   
The data structures presented in this paper, as well as most related constructions from the literature, 
have the property that $T_{c}(n)/n = O(T_{u}(n))$, allowing us to go from a partially dynamic 
to a fully dynamic data structure ``for free'' in big O notation.  

In terms of guaranteeing that the query operation solves the $(r, cr)$-near neighbor problem on the set of points $P$ 
currently inserted into the data structure, we allow a constant failure probability $\delta < 1$, 
typically around $1/2$, and omit it from our statements.
We make the standard assumption that the adversary does not have knowledge of the randomness used by the data structure.
Say we have a data structure with constant failure probability and a bound on the expected space usage. 
Then, for every positive integer $T$ we can create a collection of $O(\log T)$ independent repetitions of the data structure such that
for every sequence of $T$ operations it holds with high probability in $T$ that the space usage will never 
exceed the expectation by more than a constant factor and no query will fail.

\subsection{Model of computation}
We use the standard word RAM model as defined by Hagerup \cite{hagerup1998} with a word size of $\Theta(\log n)$ bits. 
Unless otherwise stated, we make the assumption that a point in $(X, D)$ can be stored in $d$ words and 
that the dissimilarity between two arbitrary points can be computed in $d$ operations where $d$ is a positive integer 
that corresponds to the dimension in the various well-studied settings mentioned in the main text.
Furthermore, when describing framework-based solutions to the $(r, cr)$-near neighbor problem, 
we make the assumption that we can sample, evaluate, and represent elements from $\LSF$ and $\LSH$ 
with neglible error using space and time $dn^{o(1)}$.

Many of the LSH and LSF families rely on random samples from the standard normal distribution.
We will ignore potential problems resulting from rounding due to the fact that our model only supports finite precision arithmetic.
This approach is standard in the literature and can be justified by noting that the error introduced by rounding is neglible.
Furthermore, there exists small pseudorandom standard normal distributions that support sampling 
using only few uniformly distributed bits as noted by Charikar \cite{charikar2002}. 
In much of the related literature the model of computation is left unspecified and statements 
about the complexity of solutions to the $(r, cr)$-near neighbor problem are usually made with respect to particular operations 
such as the hash function computations, distance computations, etc., leaving out other details \cite{indyk1998, har-peled2012}.  
\bibliography{filters}

\begin{thebibliography}{10}

\bibitem{aggarwal2001}
C.~Aggarwal, D.~A. Keim, and A.~Hinneburg.
\newblock On the surprising behaviour of distance metrics in high dimensional
  space.
\newblock In {\em Proc. {ICDT} '01}, pages 420--434, 2001.

\bibitem{alman2015}
J.~Alman and R.~Williams.
\newblock Probabilistic polynomials and hamming nearest neighbors.
\newblock In {\em Proc. {FOCS} '15}, pages 136--150, 2015.

\bibitem{andoni2009}
A.~Andoni.
\newblock {\em Nearest neighbor search: the old, the new, and the impossible}.
\newblock PhD thesis, MIT, 2009.

\bibitem{andoni2006}
A.~Andoni and P.~Indyk.
\newblock Near-optimal hashing algorithms for approximate nearest neighbor in
  high dimensions.
\newblock In {\em Proc. {FOCS} '06}, pages 459--468, 2006.

\bibitem{andoni2008}
A.~Andoni and P.~Indyk.
\newblock Near-optimal hashing algorithms for approximate nearest neighbor in
  high dimensions.
\newblock {\em Commun. {ACM}}, 51(1):117--122, 2008.

\bibitem{andoni2015}
A.~Andoni, P.~Indyk, T.~Laarhoven, I.~Razenshteyn, and L.~Schmidt.
\newblock Practical and optimal lsh for angular distance.
\newblock In {\em Proc. {NIPS} '15}, pages 1225--1233, 2015.

\bibitem{andoni2014}
A.~Andoni, P.~Indyk, H.~L. Nguyen, and I.~P. Razenshteyn.
\newblock Beyond locality-sensitive hashing.
\newblock In {\em Proc. {SODA} '14}, pages 1018--1028, 2014.

\bibitem{andoni2006lower}
A.~Andoni, P.~Indyk, and M.~Patrascu.
\newblock On the optimality of the dimensionality reduction method.
\newblock In {\em Proc. {FOCS} '06}, pages 449--458, 2006.

\bibitem{andoni2016}
A.~Andoni, T.~Laarhoven, I.~P. Razenshteyn, and E.~Waingarten.
\newblock Lower bounds on time-space trade-offs for approximate near neighbors.
\newblock {\em CoRR}, abs/1605.02701, 2016.

\bibitem{andoni2016optimal}
A.~Andoni, T.~Laarhoven, I.~P. Razenshteyn, and E.~Waingarten.
\newblock Optimal hashing-based time-space trade-offs for approximate near
  neighbors.
\newblock {\em CoRR}, abs/1608.03580, 2016.

\bibitem{andoni2015data}
A.~Andoni and I.~Razenshteyn.
\newblock Optimal data-dependent hashing for approximate near neighbors.
\newblock In {\em Proc. {STOC} '15}, pages 793--801, 2015.

\bibitem{andoni2015lower}
A.~Andoni and I.~P. Razenshteyn.
\newblock Tight lower bounds for data-dependent locality-sensitive hashing.
\newblock {\em CoRR}, abs/1507.04299, 2015.

\bibitem{andoni2016datalower}
A.~Andoni and I.~Razensteyn.
\newblock Tight lower bounds for data-dependent locality-sensitive hashing.
\newblock In {\em Proc. {SoCG} '16}, pages 9:1--9:11, 2016.

\bibitem{becker2016}
A.~Becker, L.~Ducas, N.~Gama, and T.~Laarhoven.
\newblock New directions in nearest neighbor searching with applications to
  lattice sieving.
\newblock In {\em Proc. {SODA} '16}, pages 10--24, 2016.

\bibitem{benyamini1998}
Y.~Benyamini and J.~Lindenstrauss.
\newblock {\em Geometric nonlinear functional analysis}, volume~48.
\newblock American Mathematical Soc., Providence, Rhode Island, 1998.

\bibitem{broder1998}
A.~Z. Broder, M.~Charikar, A.~M. Frieze, and M.~Mitzenmacher.
\newblock Min-wise independent permutations.
\newblock In {\em Proc. {STOC} '98}, pages 327--336, 1998.

\bibitem{chambers1976}
J.~M. Chambers, C.~L. Mallows, and B.~W. Stuck.
\newblock A method for simulating stable random variables.
\newblock {\em Jour. Am. Stat. Assoc.}, 71(354):340--344, 1976.

\bibitem{charikar2002}
M.~Charikar.
\newblock Similarity estimation techniques from rounding algorithms.
\newblock In {\em Proc. {STOC} '02}, pages 380--388, 2002.

\bibitem{chierichetti2015}
F.~Chierichetti and R.~Kumar.
\newblock Lsh-preserving functions and their applications.
\newblock {\em J. {ACM}}, 62(5):33, 2015.

\bibitem{datar2004}
M.~Datar, N.~Immorlica, P.~Indyk, and V.~S. Mirrokni.
\newblock Locality-sensitive hashing scheme based on p-stable distributions.
\newblock In {\em Proc. {SOCG} '04}, pages 253--262, 2004.

\bibitem{berg2008}
M.~de~Berg, M.~van Kreveld, M.~Overmars, and O.~C. Schwarzkopf.
\newblock {\em Computational geometry}.
\newblock Springer, Berlin, third edition, 2008.

\bibitem{dubiner2010}
M.~Dubiner.
\newblock Bucketing coding and information theory for the statistical
  high-dimensional nearest-neighbor problem.
\newblock {\em {IEEE} Trans. Inf. Theory}, 56(8):4166--4179, 2010.

\bibitem{hagerup1998}
T.~Hagerup.
\newblock Sorting and searching on the word {RAM}.
\newblock In {\em Proc. {STACS} '98}, pages 366--398, 1998.

\bibitem{har-peled2012}
S.~Har-Peled, P.~Indyk, and R.~Motwani.
\newblock Approximate nearest neighbor: Towards removing the curse of
  dimensionality.
\newblock {\em Theory comp.}, 8(1):321--350, 2012.

\bibitem{hoeffding1963}
W.~Hoeffding.
\newblock Probability inequalities for sums of bounded random variables.
\newblock {\em Jour. Am. Stat. Assoc.}, 58(301):13--30, 1963.

\bibitem{indyk1998}
P.~Indyk and R.~Motwani.
\newblock Approximate nearest neighbors: towards removing the curse of
  dimensionality.
\newblock In {\em Proc. {STOC} '98}, pages 604--613, 1998.

\bibitem{indyk2004}
Piotr Indyk.
\newblock Nearest neighbors in high-dimensional spaces.
\newblock In {\em Handbook of Discrete and Computational Geometry, Second
  Edition.}, pages 877--892. Chapman and Hall/{CRC}, 2004.

\bibitem{buhrman1980}
R.~Kaas and J.~M. Buhrman.
\newblock Mean, median and mode in binomial distributions.
\newblock {\em Statistica Neerlandica}, 34(1):13--18, 1980.

\bibitem{kahn1988}
J.~Kahn, G.~Kalai, and N.~Linial.
\newblock The influence of variables on boolean functions (extended abstract).
\newblock In {\em Proc. {FOCS} '88}, pages 68--80, 1988.

\bibitem{kapralov2015}
M.~Kapralov.
\newblock Smooth tradeoffs between insert and query complexity in nearest
  neighbor search.
\newblock In {\em Proc. {PODS} '15}, pages 329--342, 2015.

\bibitem{kushilevitz2000}
E.~Kushilevitz, R.~Ostrovsky, and Y.~Rabani.
\newblock Efficient search for approximate nearest neighbor in high dimensional
  spaces.
\newblock {\em {SIAM} J. Comput.}, 30(2):457--474, 2000.

\bibitem{laarhoven2015}
T.~Laarhoven.
\newblock Tradeoffs for nearest neighbors on the sphere.
\newblock {\em CoRR}, abs/1511.07527, 2015.

\bibitem{levy1925}
P.~L{\'e}vy.
\newblock {\em Calcul des probabilit{\'e}s}, volume~9.
\newblock Gauthier-Villars, Paris, 1925.

\bibitem{lu2009}
D.~Lu and W.~V. Li.
\newblock A note on multivariate gaussian estimates.
\newblock {\em Journal of Mathematical Analysis and Applications},
  354(2):704--707, 2009.

\bibitem{lukacs1970}
E.~Lukacs.
\newblock {\em Characteristic Functions}.
\newblock Griffin, London, second edition, 1970.

\bibitem{lv2007}
Q.~Lv, W.~Josephson, Z.~Wang, M.~Charikar, and K.~Li.
\newblock Multi-probe {LSH:} efficient indexing for high-dimensional similarity
  search.
\newblock In {\em Proc. {VLDB} '07}, pages 950--961, 2007.

\bibitem{marcinkiewicz1939}
J.~Marcinkiewicz.
\newblock Sur une propriete de la loi de gauss.
\newblock {\em Mathematische Zeitschrift}, 44(1):612--618, 1939.

\bibitem{minsky1969}
M.~Minsky and S.~Papert.
\newblock {\em Perceptrons}.
\newblock MIT Press, Cambridge, MA, 1969.

\bibitem{motwani2007}
R.~Motwani, A.~Naor, and R.~Panigrahy.
\newblock Lower bounds on locality sensitive hashing.
\newblock {\em {SIAM} J. Discrete Math.}, 21(4):930--935, 2007.

\bibitem{nguyen2014}
H.~L. Nguyen.
\newblock {\em Algorithms for High Dimensional Data}.
\newblock PhD thesis, Princeton, 2014.

\bibitem{od2014}
R.~O'Donnell.
\newblock {\em Analysis of Boolean Functions}.
\newblock Cambridge University Press, 2014.

\bibitem{overmars1981}
M.~H. Overmars and J.~van Leeuwen.
\newblock Worst-case optimal insertion and deletion methods for decomposable
  searching problems.
\newblock {\em Information Processing Letters}, 12(4):168--173, 1981.

\bibitem{odonnell2014}
R.~O’Donnell, Y.~Wu, and Y.~Zhou.
\newblock Optimal lower bounds for locality-sensitive hashing (except when q is
  tiny).
\newblock {\em ACM Transactions on Computation Theory (TOCT)}, 6(1):5, 2014.

\bibitem{panigrahy2006}
R.~Panigrahy.
\newblock Entropy based nearest neighbor search in high dimensions.
\newblock In {\em Proc. {SODA} '06}, pages 1186--1195, 2006.

\bibitem{panigrahy2008}
R.~Panigrahy, K.~Talwar, and U.~Wieder.
\newblock A geometric approach to lower bounds for approximate near-neighbor
  search and partial match.
\newblock In {\em Proc. {FOCS} '08}.

\bibitem{panigrahy2010}
R.~Panigrahy, K.~Talwar, and U.~Wieder.
\newblock Lower bounds on near neighbor search via metric expansion.
\newblock In {\em Proc. {FOCS} '10}, pages 805--814, 2010.

\bibitem{polya1949}
G.~P{\'o}lya.
\newblock Remarks on characteristic functions.
\newblock In {\em Proc. Berkeley Symposium on Mathematical Statistics and
  Probability 1945-1946}, pages 115--123, 1949.

\bibitem{rahimi2007}
A.~Rahimi and B.~Recht.
\newblock Random features for large-scale kernel machines.
\newblock In {\em Proc. {NIPS} '07}, pages 1177--1184, 2007.

\bibitem{rudin1990}
W.~Rudin.
\newblock {\em Fourier Analysis on Groups}.
\newblock Wiley, New York, 1990.

\bibitem{savage1962}
I.~R. Savage.
\newblock Mill's ratio for multivariate normal distributions.
\newblock {\em Jour. Res. {NBS} Math. Sci.}, 66(3):93--96, 1962.

\bibitem{scholkopf2002}
B.~Sch{\"o}lkopf and A.~J. Smola.
\newblock {\em Learning with Kernels}.
\newblock The {MIT} Press, Cambridge, Massachusetts, 2002.

\bibitem{szarek1999}
S.~J. Szarek and E.~Werner.
\newblock A nonsymmetric correlation inequality for gaussian measure.
\newblock {\em Journal of multivariate analysis}, 68(2):193--211, 1999.

\bibitem{ushakov1999}
N.~G. Ushakov.
\newblock {\em Selected Topics in Characteristic Functions}.
\newblock VSP, Utrecht, The Netherlands, 1999.

\bibitem{valiant2015}
Gregory Valiant.
\newblock Finding correlations in subquadratic time, with applications to
  learning parities and the closest pair problem.
\newblock {\em J. {ACM}}, 62(2):13, 2015.

\bibitem{wang2014}
J.~Wang, H.~T. Shen, J.~Song, and J.~Ji.
\newblock Hashing for similarity search: {A} survey.
\newblock {\em CoRR}, abs/1408.2927, 2014.

\bibitem{weber1998}
R.~Weber, H.~Schek, and B.~Stephen.
\newblock A quantitative analysis and performance study for similarity-search
  methods in high-dimensional spaces.
\newblock In {\em Proc. VLDB '98}, pages 194--205, 1998.

\bibitem{williams2004}
Ryan Williams.
\newblock A new algorithm for optimal constraint satisfaction and its
  implications.
\newblock In {\em Proc. {ICALP} '04}, pages 1227--1237, 2004.

\bibitem{zolotarev1986}
V.~M. Zolotarev.
\newblock {\em One-dimensional stable distributions}, volume~65.
\newblock American Mathematical Soc., 1986.

\end{thebibliography}
\bibliographystyle{plain}
\end{document}